\documentclass[twoside,11pt]{article}

\usepackage{graphicx, subfigure}
\usepackage[top=1in,bottom=1in,left=1in,right=1in]{geometry}
\usepackage[sort&compress]{natbib} 
\usepackage{amsfonts, amsmath, amssymb, amsthm, constants}
\usepackage{algorithm, algorithmic}
\usepackage{enumitem}
\usepackage{mathtools}
\usepackage{MnSymbol}

\usepackage{cases}
\usepackage{color}
\definecolor{darkred}{RGB}{100,0,0}
\definecolor{darkgreen}{RGB}{0,100,0}
\definecolor{darkblue}{RGB}{0,0,150}

\usepackage{hyperref}
\hypersetup{colorlinks=true, linkcolor=darkred, citecolor=darkgreen, urlcolor=darkblue}
\usepackage{url}

\definecolor{purple}{rgb}{0.4,.1,.9}

\newtheorem{thm}{Theorem}
\newtheorem{prp}{Proposition}
\newtheorem{lem}{Lemma}

{
\theoremstyle{remark}
\newtheorem{rem}{Remark}
\newtheorem{ex}{Example}

}
{
\newtheoremstyle{propertystyle} 
    {\topsep}                    
    {\topsep}                    
    {}                   
    {}                           
    {\bfseries\itshape}                   
    {.}                          
    {.5em}                       
    {}  
\theoremstyle{propertystyle}
\newtheorem{pro}{Property}
}

\newcommand{\thmref}[1]{Theorem~\ref{thm:#1}}
\newcommand{\prpref}[1]{Proposition~\ref{prp:#1}}

\newcommand{\lemref}[1]{Lemma~\ref{lem:#1}}
\newcommand{\secref}[1]{Section~\ref{sec:#1}}
\newcommand{\figref}[1]{Figure~\ref{fig:#1}}

\newcommand{\proref}[1]{Property~\ref{pro:#1}}
\newcommand{\exref}[1]{Example~\ref{ex:#1}}
\newcommand{\remref}[1]{Remark~\ref{rem:#1}}

\def\beq{\begin{equation}}
\def\eeq{\end{equation}}
\def\beqn{\begin{eqnarray*}}
\def\eeqn{\end{eqnarray*}}
\def\bitem{\begin{itemize}}
\def\eitem{\end{itemize}}
\def\benum{\begin{enumerate}}
\def\eenum{\end{enumerate}}
\def\bmult{\begin{multline*}}
\def\emult{\end{multline*}}
\def\bcenter{\begin{center}}
\def\ecenter{\end{center}}
\def\bsplit{\begin{split}}
\def\esplit{\end{split}}

\def\cA{\mathcal{A}}
\def\cB{\mathcal{B}}

\def\cS{\mathcal{S}}

\def\cU{\mathcal{U}}
\def\cV{\mathcal{V}}
\def\cW{\mathcal{W}}
\def\cX{\mathcal{X}}

\def\bb{\mathbf{b}}

\def\bbB{\mathbb{B}}

\def\bbN{\mathbb{N}}

\def\bbR{\mathbb{R}}

\newcommand{\<}{\langle}
\renewcommand{\>}{\rangle}

\def\eps{\varepsilon}

\usepackage{dsfont}

\def\S{\cS}

\DeclareMathOperator{\curv}{curv}
\def\rad{r}
\def\comp{\mathsf{c}}
\def\H{\mathsf{H}}
\def\origin{\textsc{o}}
\def\aa{a}
\def\bb{b}
\def\cc{c}
\def\k{\textsc{k}}

\begin{document}
\thispagestyle{empty}

\title{Unconstrained and Curvature-Constrained Shortest-Path Distances and their Approximation}
\author{Ery Arias-Castro\footnote{\ University of California, San Diego, USA --- \url{http://math.ucsd.edu/\~eariasca/}}
\and
Thibaut Le Gouic\footnote{\ \'Ecole Centrale de Marseille, France --- \url{http://thibaut.le-gouic.perso.centrale-marseille.fr/}}}
\date{}

\maketitle

\begin{abstract}
We study shortest paths and their distances on a subset of a Euclidean space, and their approximation by their equivalents in a neighborhood graph defined on a sample from that subset.  In particular, we recover and extend the results of \cite{bernstein2000graph}.  We do the same with curvature-constrained shortest paths and their distances, establishing what we believe are the first approximation bounds for them.
\end{abstract}

\tableofcontents

\section{Introduction} \label{sec:intro}
Finding a shortest path between two points is a fundamental problem in the general area of optimization with applications ranging from computing the shortest route between two physical locations on a road network \citep{delling2009engineering} to path planning in robotics \citep{lavalle2006planning}.  
In machine learning, shortest paths are at the core of the {\sf Isomap} algorithm for manifold learning \citep{silva2002global,Tenenbaum00ISOmap} and the {\sf MDS-MAP} algorithm \citep{shang2003localization,shang2004improved} for multidimensional scaling in the presence of missing distances.
(See also the work of \cite{kruskal1980designing}.)

In this paper we study shortest paths and shortest path distances on a subset $\S$ of a Euclidean space.  
The resulting mathematical model is closely related to sampling-based motion planning in robotics and is directly relevant to machine learning tasks such as manifold learning and manifold clustering.
Our main interest lies in how well such paths and distances are approximated by their equivalents in a neighborhood graph built on a sample of points from $\S$.  

An important concept in machine learning where data points are available (typically in a Euclidean space), a neighborhood graph based on these points is a graph with nodes indexing the points themselves and edges between two nodes when their corresponding points are within a certain distance.  An edge is weighed by the Euclidean distance between the two underlying points.
(Other variants exist.)  
The shortest paths in the neighborhood graph thus offer a natural discrete analog to the shortest paths on the set, which are continuous by nature.  
This correspondence has algorithmic implications: one is instinctively led to computing shortest paths in the graph to produce estimates for shortest paths on $\S$.  
This is exactly what {\sf Isomap} does \citep{silva2002global}.  
In robotics, sampling-based motion planning algorithms also rely on sampling the environment (here the surface) to discretize the optimization problem \citep{thrun2005probabilistic}.

\subsection{Shortest paths}
Shortest paths are, of course, well-studied objects in mathematics, particularly in metric geometry \citep{burago2001course}.
They have also been the object of intensive study in robotics \citep{latombe2012robot}.  
The approximation of shortest path distances on the surface by shortest path distances in the graph was established by \cite{bernstein2000graph} with a view towards providing theoretical guarantees for {\sf Isomap}.  This was done in the context of a geodesically convex surface.
The same sort of approximation has also been considered in robotics; see, for example, in \citep{karaman2011sampling, kavraki1998analysis, janson2015fast, lavalle2001randomized}, and references therein, where the context is that of a Euclidean domain with holes representing the obstacles that the robot needs to avoid.  
Note that, in this literature, the neighborhood graph is modified to only include collision-free edges.

In \secref{unconstrained} we establish some basic results on shortest paths and the corresponding metric.  We then revisit the results of \cite{bernstein2000graph} on the approximation of this metric with the pseudometric on a neighborhood graph, making a connection with the seminal work of \cite{dubins1957curves}.  We also show that it is possible to approximate shortest paths on the surface with shortest paths in a neighborhood graph under much milder regularity assumptions.

\subsection{Curvature-constrained shortest paths}
We also consider curvature-constrained shortest paths and the corresponding distances.  These have long been considered in robotics to model settings where the robot has limited turning radius.  Theoretical results date back at least to the seminal work of \cite{dubins1957curves}.  See also \citep{boissonnat1992shortest, reeds1990optimal}.  
Still in robotics, sampling-based motion planning algorithms designed to satisfy differential motion constraints (including kinematic/dynamical constraints) are studied, for example, in 
\citep{li2015sparse, karaman2010optimal, schmerling2015optimal, schmerling2015optimal-drift}.  
Importantly, in the present work, no constraints are placed on the initial and final orientations.  Another difference with this literature, where motion is most typically in a Euclidean domain, is that we consider a surface that may be curved, and this forces the shortest paths on to satisfy some nontrivial curvature constraint.
  
In machine learning, angle and curvature-constrained paths have been recently considered for the task of surface learning where the surface may self-intersect \citep{babaeian2015nonlinear} and for the task of multi-surface clustering \citep{babaeian2015multiple}.

In \secref{constrained} we derive some theory on curvature-constrained shortest paths and the corresponding (semi-)metric, and in particular establish bounds on approximations by curvature-constrained shortest paths in a neighborhood graph.  This requires a notion of discrete curvature applicable to polygonal lines, which we describe in \eqref{curv-def}.  

\section{Preliminaries}
In this section we set most of the notation for the reminder of the paper, introduce some fundamental concepts in metric geometry, and also list some basic results that will be used later on.

\paragraph{Vectors}
For two vectors $u,v \in \bbR^D$, their inner product is denoted $\<u,v\>$ and their angle is defined as $\angle (u,v) = \cos^{-1}(\<u,v\>/\|u\| \|v\|) \in [0,\pi]$, where $\|\cdot\|$ denotes the Euclidean norm.  We also define $u \wedge v$ as their wedge product.  Recall that $u \wedge v$ is a 2-vector, and the space of 2-vectors can be endowed with an inner product, and the resulting norm --- also denoted by $\|\cdot\|$ --- satisfies $\|u \wedge v\| = \|u\| \|v\| |\sin \angle(u,v)|$, which is also the area of the parallelogram defined by $u$ and $v$.

\paragraph{Sets}
For $x \in \bbR^D$ and $r > 0$, let $B(x,r)$ denote the open ball of $\bbR^D$ with center $x$ and radius $r$.
For $\cA, \cB \subset \bbR^D$, let $\cA \oplus \cB = \{a + b : a \in \cA, b \in \cB\}$, their Minkowski sum.  In particular, $\cA \oplus B(0, r)$ is the $r$-tubular neighborhood of $\cA$.

\begin{lem} \label{lem:compact-in-open}
Consider two subsets $\cA \subset \cB$ of a Euclidean space, with $\cA$ compact and $\cB$ open.  Then there is $r > 0$ such that $\cA \oplus B(0, r) \subset \cB$.  
\end{lem}

\begin{proof}
Suppose the statement is not true.  In that case, for $m \ge 1$ integer, take $x_m \in \cA \oplus B(0, 1/m) \setminus \cB$.  Since $(x_m)$ is bounded, we may assume WLOG that it converges to some $x$.  
(Here and elsewhere, $(x_m)$ is shorthand for the sequence $x_1, x_2, \dots$, and when we write $x_m \to x$ we mean that $(x_m)$ converges to $x$.)
Clearly $x \in \cA$, since $x \in \cA \oplus B(0, 1/m)$ for all $m\ge1$ and $\cA$ is compact.  This implies that $x \in \cB$, and since $\cB$ is open, there is $r > 0$ such that $B(x,r) \subset \cB$.  Since $x_m \to x$, for $m$ large enough, $x_m \in B(x,r)$, implying $x_m \in \cB$, which is a contradiction.
\end{proof}

For two sets $\cA$ and $\cB$, define their Hausdorff distance
\beq \label{hausdorff}
\H(\cA, \cB) = \max(\H(\cA \mid \cB), \H(\cB \mid \cA)), \qquad \H(\cA \mid \cB) := \sup_{a \in \cA} \inf_{b \in \cB} \|a - b\|.
\eeq
By convention, we set $\H(\emptyset, \cA) = \infty$ for any $\cA$.
Note that 
\[
\H(\cA \mid \cB) = \inf\big\{h >0 : \cA \subset \cB \oplus B(0,h)\big\}.
\]  

%

\paragraph{Curves}
A curve in $\bbR^D$ is a continuous function $\gamma: I \to \bbR^D$ on an interval $I \subset \bbR$.  
We will often identify the function $\gamma$ with its image $\gamma(I)$. 
We say that a sequence of curves $\gamma_m$ converges uniformly to a curve $\gamma$ if these curves can be parameterized in such a way that the convergence as functions is uniform; see \cite[Def 2.5.13]{burago2001course}.  

We note that, for two parameterized curves $\gamma: I \to \bbR^D$ and $\zeta: I \to \bbR^D$,
\beq \label{hausdorff-unif}
\H(\gamma, \zeta) \le \sup_{t \in I} \|\gamma(t) - \zeta(t)\|,
\eeq
and in particular, for curves, uniform convergence implies convergence in Hausdorff metric.

\paragraph{Length}
For $Z = (z_1, \dots, z_m) \subset \bbR^D$, let $\Lambda(Z) = \sum_i \|z_i - z_{i+1}\|$, which is the length of the polygonal line defined by $Z$. 
This definition is extended to general curves in the usual way \cite[Def 2.3.1]{burago2001course}: the length of a curve $\gamma : I \to \bbR^D$ is 
\beq\label{length}
\Lambda(\gamma) = \sup \sum_j \|\gamma(t_{j+1}) - \gamma(t_j)\|,
\eeq
where the supremum is over all increasing sequences $(t_j) \subset I$. 
%
Any curve with finite length admits a unit-speed parameterization \cite[Def~2.5.7, Prop~2.5.9]{burago2001course}, meaning that for any such curve $\gamma: I \to \bbR^D$ there is an interval $J \subset \bbR$ and a continuous function $\nu : J \to \bbR^D$ such that $\nu(J) = \gamma(I)$ and $\Lambda(\nu([a,b])) = b-a$ for all $a < b$ in $J$.  All curves will be assumed to be unit-speed (i.e., parameterized by arc length) by default.  Note that a unit-speed curve is differentiable almost everywhere with unit norm derivative, meaning $\|\dot\gamma(t)\| = 1$ for almost all $t \in I$; in particular, such a curve is 1-Lipschitz, where we say that a function $f: \Omega \subset \bbR^k \to \bbR^l$ is $L$-Lipschitz for some $L > 0$ if $\|f(x) - f(y)\| \le L \|x - y\|$ for all $x,y \in \Omega$.  

\begin{lem} \label{lem:Lip-length}
For any curve $\gamma$ and any $L$-Lipschitz function $f$, $\Lambda(f \circ \gamma) \le L \Lambda(\gamma)$.
\end{lem}
\begin{proof}
Consider a parameterization $\gamma: I \to \bbR^D$.  Then for any increasing sequence $(t_j) \subset I$,
\[
\Lambda(f \circ \gamma)
\le \sum_j \|f \circ \gamma(t_{j+1}) - f \circ \gamma(t_j)\|
\le L \sum_j \|\gamma(t_{j+1}) - \gamma(t_j)\|,
\]
and by taking the supremum of such sequences, we obtain the result, since the right-hand side becomes $L \Lambda(\gamma)$.
\end{proof}

\begin{lem} \label{lem:projection-ball}
Suppose $\gamma$ is a curve and $B$ is a closed ball such that $\gamma \cap B = \emptyset$.  If $\zeta$ denotes the metric projection of $\gamma$ onto $B$, then $\Lambda(\zeta) < \Lambda(\gamma)$.
\end{lem}
\begin{proof}
Assume WLOG that $B$ is the closed unit ball and consider a unit-speed parameterization $\gamma: [0,\ell] \to \bbR^D$.  
Let $f$ denote the metric projection onto $B$, which has the simple expression $f(x) = x/\|x\|$ for $x \notin B$ (while, of course, $f(x) = x$ for $x \in B$).  By continuity, there is $h > 1$ and a subinterval $[a,b] \subset I$, with $a < b$, such that $\|\gamma(s)\| \ge h$ for all $s \in [a,b]$.  
For $x \ne 0$, the differential of $f$ at $x$, denoted $D_x f$, is equal to $\frac1{\|x\|} ({\rm Id} - x x^\top/\|x\|^2)$, where ${\rm Id}$ denotes the identity linear function.  Hence, $D_x f$ has operator norm equal to $1/\|x\|$.
By Taylor's theorem, we thus have 
\[
\|f(x) - f(y)\| \le \frac{\|x - y\|}{\min(\|x\|, \|y\|)}, 
\]
for all $x, y$ such that the line segment joining $x$ and $y$ does not contain the origin, and this extends to all $x,y \ne 0$ by continuity.  In particular, $f$ is $(1/h)$-Lipschitz on $\{x : \|x\| \ge h\}$.  
Then, by \lemref{Lip-length},
\begin{align*}
\Lambda(\zeta)
&= \Lambda(\zeta([0,a]) + \Lambda(\zeta([a,b]) + \Lambda(\zeta([b,\ell]) \\
&\le \Lambda(\gamma([0,a]) + (1/h) \Lambda(\gamma([a,b]) + \Lambda(\gamma([b,\ell]) \\
&< \Lambda(\gamma([0,a]) + \Lambda(\gamma([a,b]) + \Lambda(\gamma([b,\ell]) 
= \Lambda(\gamma),
\end{align*}
by the fact that $1/h < 1$ and $\Lambda(\gamma([a,b]) = b -a > 0$ by construction.
\end{proof}

\paragraph{Curvature}
We say that a unit-speed curve $\gamma$ has curvature bounded by $\kappa$ if it is differentiable and its derivative is $\kappa$-Lipschitz.  
Assuming the curve $\gamma$ is twice differentiable at $t$, its curvature at $t$ is defined as
\beq \label{curv-gamma}
\curv(\gamma,t) = \frac{\|\dot\gamma(t) \wedge \ddot\gamma(t)\|}{\|\dot\gamma(t)\|^3}.
\eeq
In that case, $\gamma$ has curvature bounded by $\kappa$ if and only if $\sup_t \curv(\gamma,t) \le \kappa$.  


\section{Shortest paths and their approximation} \label{sec:unconstrained}

In this section, we consider the intrinsic metric on a subset and its approximation by a pseudo-metric defined based on a neighborhood graph built on a finite sample of points from the surface.
In \secref{intrinsic-metric} we define the intrinsic metric on a given subset, and list a few of its properties.
In \secref{pseudo-metric} we define the notion of neighborhood graph and a pseudo-metric based on shortest path distances in that graph.
In \secref{approximation} we show that this pseudo-metric can be used to approximate the intrinsic metric on a subset.  We discuss the results obtained by \cite{bernstein2000graph} and make a connection with the classical work of \cite{dubins1957curves}.

\subsection{The intrinsic metric on a surface}
\label{sec:intrinsic-metric}
The intrinsic metric on a set $\S$ is the metric inherited from the ambient space $\bbR^D$.  For $x, x' \in \S$, it is defined as
\beq \label{delta_S}
\delta_\S(x, x') = \inf\big\{a : \exists \gamma: [0,a] \to \S, \text{ 1-Lipschitz, with } \gamma(0) = x \text{ and } \gamma(a) = x'\big\}.
\eeq

When $\delta_\S(x,x') < \infty$ for all $x,x' \in \S$, we say that $\S$ is path-connected \citep[Sec 2.5]{waldmann2014topology}.
A lot is known about this type of metric \citep{burago2001course}.  In particular, when $\S$ is a smooth submanifold, this is the Riemannian metric induced by the ambient space \cite[Sec 5.1.3]{burago2001course}.

\begin{lem} \label{lem:delta_S}
If $\S \subset \bbR^D$ is closed and $x,x' \in \S$ are such that $\delta_\S(x,x') < \infty$, the infimum in \eqref{delta_S} is attained.  If $\S \subset \bbR^D$ is open and connected, $\delta_\S(x,x') < \infty$ for all  $x,x' \in \S$.
\end{lem}

\begin{proof}
For the first part, we refer the reader to the proof of \lemref{delta_kappa}.
For the second part, every connected open set in a Euclidean space is also path-connected \citep[Ex~2.5.13]{waldmann2014topology}.    
\end{proof}

The intrinsic metric on $\S \subset \bbR^D$ is in general different from the ambient Euclidean metric inherited from $\bbR^D$.  The two coincide only when $\S$ is convex.  However, it is true that $\|x-x'\| \le \delta_\S(x,x')$ for all $x,x' \in \S$, and in particular this implies that the ambient topology is always at least as fine as the intrinsic topology.  But there are cases where the two topologies differ.  

\begin{ex}[A set with infinite intrinsic diameter and finite ambient diameter]  \label{ex:spiral}
Consider a closed spiral with infinite length, for example defined as $\S = \bar\Gamma$, where $\Gamma(t) := (\rho(t) \cos(t), \rho(t) \sin(t))$ for $t \ge 0$ and $\rho : [0, \infty) \to (0,1]$ one-to-one (decreasing) and such that $\int_0^\infty \rho(t) {\rm d}t = \infty$.  The resulting set $\S$ is compact in $\bbR^2$, but unbounded for its intrinsic metric since $\delta_\S(\origin, x) = \infty$ for all $x \in \S$.  ($\origin$ denotes the origin.)  In particular, if $t_k \to 0$ and $x_k = \Gamma(t_k)$, then $x_k \to \origin$ in the ambient topology, while $x_k \not\to \origin$ in the intrinsic topology.
Suppose we now thicken the spiral and redefine $\S = \overline{\bigcup_{t \ge 0} \Gamma(t) \oplus B(\origin, \nu(t))}$ where $\nu : [0,\infty) \to (0,1]$ is decreasing and such that $\nu(t) + \nu(t+2\pi) < \rho(t) - \rho(t+2\pi)$.
In that case, $\S$ is the closure of its interior and, assuming $\rho$ and $\nu$ are $C^\infty$, $\partial \S$ is $C^\infty$ except at the origin~\origin.  And still, $\S$ has infinite intrinsic diameter.  
\end{ex}


\begin{ex}[A set with finite intrinsic diameter having different intrinsic and ambient topologies] \label{ex:slices}
Let $\Gamma(t) = (\rho(t) \cos(t), \rho(t) \sin(t))$ with $\rho: [0, 2\pi) \to (0,1]$ continuous and strictly decreasing and satisfying $\lim_{t\to 2\pi} \rho(t) = 0$.  Consider a strictly increasing sequence $t_0 = 0 < t_1 < t_2 < \cdots$ such that $\lim_{k \to \infty} t_k = 2\pi$.  In the cone of $\bbR^2$ defined in polar coordinates by $\{(r, \theta) : t_{2k} < \theta < t_{2k+2}\}$ consider a $C^\infty$ self-avoiding path $\gamma_k$ of length 1 starting at $x_k :=\Gamma(t_{2k+1})$ and ending at the origin \origin.  Note that $\gamma_k \cap \gamma_l = \{\origin\}$ when $k \ne l$.  Define $\S = \overline{\bigcup_k \gamma_k}$.  Clearly, $\delta_S(x, \origin) \le 1$ for all $x \in \S$, so that $\S$ has intrinsic diameter bounded by 2.  By construction $(x_k)$ converges to \origin~in the ambient topology but is not even convergent in the intrinsic topology.  (If it were to converge in the intrinsic topology, the limit would have to be \origin, but $\delta_\S(x_k, \origin) = 1$ for all $k$.)
Also, as in \exref{spiral}, if we carefully thicken each $\gamma_k$, the resulting $\S$ can be made to be the closure of its interior and have $C^\infty$ border except at the origin~\origin.
\end{ex}


Having established that the ambient and intrinsic topologies need not coincide, we will mostly focus on the case where they do, which corresponds to assuming the following.
\begin{pro} \label{pro:coincide}
The intrinsic and ambient topologies coincide on $\S$.
\end{pro}

The following is well-known to the specialist.  We provide a proof for completeness, and also because similar, but more complex arguments will be used later on.

\begin{lem} \label{lem:smooth}
Any smooth submanifold of $\S \subset \bbR^D$ with empty or smooth boundary satisfies \proref{coincide}.
\end{lem}

\begin{proof}
Let $\S \subset \bbR^D$ be a smooth submanifold of dimension $d$.  Assume for contradiction that the topologies do not coincide.  Then there is $x \in \S$ and $c > 0$, and a sequence $(x_m) \in \S$ such that $\|x-x_m\| \le 1/m$ and $\delta_\S(x,x_m) \ge c$ for all $m \ge 1$.  Let $\cW \subset \bbR^D$ be an open set containing $x$ and $f : \cW \cap \S \to \cU$ be a diffeomorphism, where $\cU$ is an open subset of either $\bbR^d$ or $\bbR^{d-1} \times \bbR_+$ --- the latter if $x \in \partial \S$.  
Let $u = f(x)$ and $h > 0$ such that $\cU_0 := \bar B(u,h) \subset \cU$.
Let $W_0 = f^{-1}(\cU_0)$ and define $\lambda = \max_{w_0 \in \cW_0} \|D_{w_0} f\| < \infty$ and $\lambda^- = \max_{u_0 \in \cU_0} \|D_{u_0} f^{-1}\| < \infty$, where in this instance $\|\cdot\|$ denotes the usual operator.  
For $m$ sufficiently large we have $x_m \in \cW$, in which case we let $u_m = f(x_m)$.  For $m$ sufficiently large we also have $[u u_m] \subset \cU$, in which case we let $\gamma_m(t) = f^{-1}((1-t) u + t u_m)$ defined on $[0,1]$.  Then $\gamma_m([0,1])$ is a curve on $\S$ joining $x$ and $x_m$, of length
\[\Lambda(\gamma_m) = \int_0^1 \|\dot\gamma_m(t)\| {\rm d}t = \int_0^1 \|(D_{\gamma_m(t)} f)^{-1} (u_m - u)\| {\rm d}t \le \lambda^- \|u_m - u\| \le \lambda^- \lambda \|x_m-x\|. 
\]
This leads to a contradiction, since $\Lambda(\gamma_m) \ge \delta_\S(x,x_m) \ge c > 0$ while $\|x_m-x\| \le 1/m \to 0$.
\end{proof}

\subsection{A neighborhood graph and its metric}
\label{sec:pseudo-metric}
We approximate the intrinsic metric (shortest-path distance) on $\S$ with the metric (shortest-path distance) on a neighborhood graph based on a sample from $\S$ denoted $\cX = \{x_1, \dots, x_N\} \subset \S$.  
While such a graph has node set indexing $\cX$ --- which we take to be $\{1, \dots,N\}$ --- there are various ways of defining the edges.  
In what follows, we write $i \sim j$ when nodes $i$ and $j$ are neighbors in the graph.
We will use the following well-known variant \citep{maier2009optimal}:
\bitem \setlength{\itemsep}{0in}
\item {\em $\rad$-ball graph:}  $i \sim j$ if and only if $\|x_i - x_j\| \le \rad$.
\eitem
We weigh each edge $i \sim j$ with the Euclidean distance between $x_i$ and $x_j$, and set the weight to $\infty$ when $i \nsim j$, thus working with
\beq \label{weight}
w_r(i,j) = \begin{cases}
\|x_i - x_j\|, & \text{if } i \sim j; \\
\infty, & \text{otherwise}.
\end{cases}
\eeq

In the context of a weighted graph, we can define a path as simply a sequence of nodes, and its length is then the sum of the weights over the sequence of node pairs that defines it.  
In our context, the length of a path $(i_1, \dots, i_m)$ is thus
\[
\Lambda_r(i_1, \dots, i_m) = \sum_{j=1}^{m-1} w_r(i_j, i_{j+1}).
\]
Equivalently, this is the length of the polygonal line defined by the sequence of points $(x_{i_1}, \dots, x_{i_m})$.
The shortest-path distance between $i$ and $j$ is defined as the length of the shortest path joining $i$ and $j$, namely
\beq\label{Lambda}
\Lambda_r^*(i,j) = \min\big\{\Lambda_r(k_1, \dots, k_m): m \ge 1, k_1 = i, k_m = j\big\}.
\eeq
This is the discrete analog of the intrinsic metric on a set defined in \eqref{delta_S}.

For two sample points, $x_i, x_j \in \cX$, define 
\[\Delta_\rad(x_i, x_j) = \Lambda^*_\rad(i,j),\]
thus defining a metric on the sample $\cX$.

\begin{rem}
\label{rem:extension}
This can be extended to a pseudo-metric\footnote{ A `pseudo-metric' is like a metric except that it needs not be definite.} on the surface $\S$ as follows: for $x,x' \in \S$, define
\[
\Delta_\rad(x,x') = \min_{i \in I(x)} \min_{j \in I(x')} \ \Lambda^*_\rad(i,j),
\]
where 
\[
I(x) := \big\{i \in [N] : \|x - x_i\| = \min_{k\in[N]} \|x-x_k\|\big\},
\]
so that $I(x)$ indexes the sample points that are nearest to $x$.  ($\Delta_r$ is only a pseudo-metric on $\S$ since we may have $\Delta_r(x, x') = 0$ even when $x \ne x'$.)
\end{rem}

\subsection{Approximation}
\label{sec:approximation}
The construction of a pseudo-metric on a neighborhood graph is meant here to approximate the intrinsic metric on the surface.  The approximation results that follow are based on how dense the sample is on the surface, which we quantify using the Hausdorff distance between $\cX$ and $\S$ as sets, namely
\beq \label{eps}
\eps = \H(\S \mid \cX) = \sup_{x \in \S} \min_{i \in [N]} \|x-x_i\|.
\eeq

Comparing $\Delta_\rad(x, x')$ with $\delta_\S(x,x')$ is exactly what \cite{bernstein2000graph} did to provide theoretical guarantees for Isomap \citep{Tenenbaum00ISOmap}.
We extend their result to more general surfaces and also provide a convergence result under very mild assumptions; see \thmref{unconstrained-conv} below.

\begin{prp} \label{prp:unconstrained-ub}
Consider $\S \subset \bbR^D$ compact and a sample $\cX=\{x_1, \dots, x_N\}\subset \S$, and let $\eps = \H(\S \mid \cX)$.  For $\rad > 0$, form the corresponding $\rad$-ball graph.  
When $\eps \le \rad/4$, we have
\[
\Delta_\rad(x, x') \le (1 + 4 \eps/\rad) \delta_\S(x,x'), \quad \forall x,x' \in \cX.
\] 
\end{prp}

This was established in \cite[Th 2]{bernstein2000graph} under essentially the same conditions, but we provide a proof for completeness.

\begin{proof}
If $\|x - x'\| \le r$, then $x$ and $x'$ are direct neighbors in the graph and so
\beq
\Delta_\rad(x, x') 
= \|x -x'\|
\le \delta_\S(x,x').
\eeq
We thus turn to the case where $\|x - x'\| > r$.
Let $a = \delta_\S(x,x')$ and let $\gamma: [0,a] \to \S$ be parameterized by arc length such that $\gamma(0) = x$ and $\gamma(a) = x'$, which exists by \lemref{delta_S}.  Let $y_j = \gamma(j a/m)$ for $j = 0, \dots, m$, where $m := \lceil 2 a/\rad \rceil \ge 2$, noting that $y_0 = x$ and $y_m = x'$.  Let $x_{i_j}$ be closest to $y_j$ among the sample points, noting that $x_{i_0} = x$ and $x_{i_m} = x'$.  In particular, $\max_j \|x_{i_j} - y_j\| \le \eps$ by definition of $\eps$.
By the triangle inequality, for any $j \in \{0, \dots, m-1\}$,
\begin{align}
\| x_{i_{j+1}}-x_{i_{j}} \| 
&\le \| x_{i_{j+1}}-y_{j+1} \|+  \| y_{j+1}-y_{j} \|+  \| y_{j} - x_{i_{j}}\| \label{unconstrained-ub-proof1}\\
&\le \eps + \delta_\S(y_{j+1},y_{j}) + \eps 
= a/m + 2\eps 
\le \rad/2 + 2\eps
\le \rad, \notag
\end{align}
so that $(x_{i_0}, \dots, x_{i_m})$ forms a path in the $\rad$-ball graph.  
We then have, using the fact that $y_0 = x_{i_0} = x$ and $y_m = x_{i_m} = x'$,
\begin{align*}
\Delta_\rad(x,x') 
\le \sum_{j=0}^{m-1} \| x_{i_{j+1}}-x_{i_{j}} \| 
&\le \delta_\S(y_0,y_1) + \eps + \sum_{j=1}^{m-2} (\delta_\S(y_{j+1},y_{j}) + 2 \eps) + \delta_\S(y_{m-1},y_m) + \eps \\
&= \delta_\S(x,x') + 2 (m-1) \eps \\
&\le (1 + 4\eps/\rad) \delta_\S(x,x'),  
\end{align*}
using the fact that $m-1 \le 2 a/\rad$ with $a = \delta_\S(x,x')$.
\end{proof}

\begin{rem}
It is possible to tighten the bound in the very special case where $\S$ is convex (and in particular flat).  Indeed, a refinement of the arguments provided above lead to an error term in $(\eps/r)^2$.  We do not know if this extends to the case where $\S$ is curved beyond the the case where it is isometric to a convex set.
\end{rem}

We establish a complementary lower bound in \prpref{unconstrained-lb} below under some regularity assumptions on $\S$.  
Before doing so, we use \prpref{unconstrained-ub} to derive a qualitative result that states that, under very mild assumptions on $\S$, shortest paths in a neighborhood graph can indeed be approximated by shortest paths on $\S$.  (\prpref{unconstrained-lb} will provide a quantitative error bound for this approximation.)

\begin{thm} \label{thm:unconstrained-conv}
Consider $\S \subset \bbR^D$ compact and satisfying \proref{coincide}.
For any $\eta > 0$, there is $\eta_0 > 0$ such that the following holds.  
Consider a sample $\cX=\{x_1, \dots, x_N\}\subset \S$ and let $\eps = \H(\S \mid \cX)$.  
For $\rad > 0$, form the $\rad$-ball graph based on $\cX$.  
Take any $x,x' \in \S$ such that both $\delta_\S(x,x') < \infty$ and $\Delta_\rad(x,x') < \infty$.  If $\rad < \eta_0$ and $\eps/\rad < \eta_0$, then for every shortest path $p$ in the graph joining $x$ and $x'$, seen as a unit-speed polygonal curve in $\bbR^D$, there is a shortest path $\gamma$ on $\S$ joining $x$ and $x'$ such that $\H(p, \gamma) \le \eta$. 
\end{thm}

\begin{proof}
Fix $x,x' \in \S$ such that $a := \delta_\S(x,x') < \infty$.
We first prove that there is such an $\eta_0 > 0$, but that may depend on $x$ and $x'$.  For this, we reason by contradiction: assuming the statement is false, there exists $\eta > 0$, and sequence $(\rad_m)$ and $(\eps_m)$ such that $\rad_m = o(1)$, $\eps_m = o(\rad_m)$, and a sample $x_{m,1}, \dots x_{m,N} \in \S$ with $\sup_{x \in \S} \min_i \|x-x_{m,i}\| = \eps_m$, as well as a shortest path $p_m$ in the $\rad_m$-ball neighborhood graph joining $x$ and $x'$ with the property that $\H(p_m, \gamma) \ge \eta$ for any shortest path $\gamma$ on $\S$ joining $x$ and $x'$.
Note that $\Lambda(p_m) \le (1+4\eps_m/\rad_m) a +2\eps_m$
by \prpref{unconstrained-ub}.
Assume WLOG that $4\eps_m/\rad_m \le 1$ and $2\eps_m \le 1$, which in particular implies that $\Lambda(p_m) \le 2 a+1$ for all $m$.

It is thus possible to parameterize each $p_m$ so that it is 1-Lipschitz on $[0, 2a+1]$.  As a family of functions on $[0,2a+1]$, $\{p_m\}$ is therefore equicontinuous.  The family is also uniformly bounded by virtue of the fact that each $p_m$ starts at $x$ and as length bounded by $2a+1$.  By the Arzel\`a-Ascoli theorem, there is a subsequence of $(p_m)$ that converges uniformly as 1-Lipschitz functions on $[0, 2a+1]$ to some 1-Lipschitz function $\gamma$ defined on that same interval.  Necessarily, $\gamma(0) = x$ and $\gamma(2a+1) = x'$.  WLOG assume that this subsequence is $(p_m)$ itself.  
We then have 
\[\Lambda(\gamma) \le \liminf_m \Lambda(p_m) \le a,\] 
by the fact that $\Lambda$ is lower semi-continuous \citep[Prop 2.3.4]{burago2001course}.
Moreover, this uniform convergence as functions implies a uniform convergence as sets, specifically, $\H(p_m, \gamma) \to 0$ as $m \to \infty$, as seen in \eqref{hausdorff-unif}.

We claim that $\gamma \subset \S$.  
If not, by the fact that $\S$ is closed, there is $t,h > 0$ such that $\gamma(t) \ne \gamma(t+h)$ and $\gamma([t,t+h]) \subset \S^\comp$.  
As $m \to \infty$,
\[\Lambda(p_m([t,t+h]) \ge \|p_m(t) - p_m(t+h)\| \to \|\gamma(t) - \gamma(t+h)\| > 0.\]
Since the line segments making up $p_m$ have length bounded by $\rad_m$, $p_m([t,t+h]))$ must include at least one sample point when $\rad_m < \|\gamma(t) - \gamma(t+h)\|$.
Therefore, $p_m([t,t+h])) \cap \S \ne \emptyset$ for $m$ large enough.
Let $t_m \in [t,t+h]$ be such that $p_m(t_m) \in \S$.  
By compactness, we may assume WLOG that $t_m \to t_\infty \in [t,t+h]$.
We then have $p_m(t_m) \to \gamma(t_\infty)$ by uniform convergence, so that $\gamma(t_\infty) \in \S$ since $\S$ is closed.  So we have a contradiction.

Collecting our findings, we found a curve $\gamma \subset \S$ joining $x$ and $x'$, with $\Lambda(\gamma) \le a = \delta_\S(x,x')$ and $\H(p_m, \gamma) \to 0$ as $m \to \infty$, which contradicts our working hypothesis that $\H(p_m, \gamma) \ge \eta > 0$ for all $m$.  This proves the first part of the theorem.

We now show that one can choose $\eta_0 > 0$ that works for all $x, x' \in \S$. 
We reason by contradiction exactly as before, except that now $x,x'$ are replaced by $x_m,x'_m$, and $a$ by $a_m := \delta_\S(x_m, x'_m)$, thus all possibly changing with $m$.  By the fact that $\S$ is compact, we may assume WLOG that there are $x,x' \in \S$ such that $x_m \to x$ and $x'_m \to x'$.  By \proref{coincide}, we have that $\delta_\S$ is continuous with respect to the Euclidean metric.  In particular, $a_m \to a := \delta_\S(x,x')$.
With this we can now see that the remaining arguments are identical to those backing the first part.
\end{proof}

We now turn to proving a bound that complements \prpref{unconstrained-ub}, that is, a more quantitative (or explicit) version of \thmref{unconstrained-conv}.  
\cite{bernstein2000graph} obtain such a bound when $\S$ is a submanifold without intrinsic curvature.  Their cornerstone result is the following, which they call the Minimum Length Lemma.

\begin{lem} \label{lem:approx} \citep{bernstein2000graph}
Let $\gamma : [0,a] \to \bbR^D$ be a unit-speed curve with curvature bounded by $\kappa$.  Then $\|\gamma(t) - \gamma(s)\| \ge \frac2\kappa \sin(\kappa |t-s|/2)$ for all $s,t \in [0,a]$ such that $|t-s| \le \pi/\kappa$.
\end{lem}

Note that the result is sharp in that the inequality is an equality when $\gamma$ is a piece of a circle of radius $1/\kappa$.
Of course, we also have $\|\gamma(t) - \gamma(s)\| \le |t-s|$ for all $s,t \in [0,a]$, since $\|\dot\gamma\|_\infty = 1$.

While \cite{bernstein2000graph} prove this result from scratch, a very short proof of a slightly weaker bound follows from a result in the pioneering work of \cite{dubins1957curves} on shortest paths with curvature constraints.  Indeed, in the setting of \lemref{approx}, let $c$ denote a unit-speed parametrization of a circle of radius $1/\kappa$. 
\cite[Prop~2]{dubins1957curves} says that $\<\dot\gamma(s), \dot\gamma(u)\> \ge \<\dot c(s), \dot c(u)\>$ when $|s-u| \le \pi/\kappa$, which leads to 
\begin{align*}
\|\gamma(t) - \gamma(s)\| 
\ge \<\dot\gamma(s), \gamma(t) - \gamma(s)\> 
&= \int_s^t \<\dot\gamma(s), \dot\gamma(u)\> {\rm d}u \\
&\ge \int_s^t \<\dot c(s), \dot c(u)\> {\rm d}u = \<\dot c(s), c(t) - c(s)\> = \frac1\kappa \sin(\kappa (t-s)),
\end{align*} 
when $0 \le t-s \le \pi/\kappa$.
(The first inequality is due to the fact that $\|\dot\gamma(s)\| = 1$.)

Using either \lemref{approx} or this weaker bound, it is straightforward to obtain a useful comparison between the intrinsic metric and the ambient Euclidean metric, locally.
Note that we still follow the footsteps of \cite{bernstein2000graph}.
The core assumption is the following.

\begin{pro}\label{pro:finite}
The shortest paths on $\S$ have curvature bounded by $\kappa$.
\end{pro}
This is true when $\S$ is sufficiently smooth.  See \lemref{finite} further down.

\begin{lem} \label{lem:approx-S}
Suppose $\S \subset \bbR^D$ is compact and satisfies \proref{coincide} and \proref{finite}.
Then, for any $x,x' \in \S$ such that $\delta_\S(x,x') \le \pi/\kappa$,
\beq \label{approx-S1}
\delta_\S(x,x') \max\big(\tfrac2\pi, 1 - \tfrac{\kappa^2}{24} \delta_\S(x,x')^2\big) \le \|x -x'\| \le \delta_\S(x,x').
\eeq
Moreover, there is $\tau > 0$ depending on $\S$ such that, for all $x,x' \in \S$, if $\|x-x'\| \le \tau$ then 
\beq \label{approx-S2}
\|x-x'\| \le \delta_\S(x,x') \le \|x-x'\|\min\big(\tfrac\pi2, 1 +  c_0 \kappa^2 \|x-x'\|^2\big),
\eeq
where $c_0$ is a universal constant that can be taken to be $= \pi^2/50$.
\end{lem}

\begin{proof}
We start with \eqref{approx-S1}, where only the lower bound is nontrivial.
Take $x,x' \in \S$ and let $a = \delta_\S(x,x') \le \pi/\kappa$.
Let $\gamma : [0,a] \to \S$ be a unit-speed shortest path on $\S$ joining $x = \gamma(0)$ and $x' = \gamma(a)$.  
By \proref{finite}, $\gamma$ has curvature bounded by $\kappa$.  Knowing that, we apply \lemref{approx} to get 
\[\|x-x'\| = \|\gamma(0)-\gamma(a)\| \ge \tfrac2\kappa \sin(\kappa a/2) \ge \max\big(a - \tfrac{\kappa^2}{24} a^3, \tfrac2\pi a\big).\]

We now turn to \eqref{approx-S2}, where only the upper bound remains to be proved.
By \eqref{approx-S1}, if $\delta_\S(x,x') \le \pi/\kappa$, then $\delta_\S(x,x') \le \frac\pi2 \|x-x'\|$ and also 
\begin{align*}
\delta_\S(x,x') 
&\le \|x-x'\|/(1-\tfrac{\kappa^2}{24} \delta_\S(x,x')^2) \\
&\le \|x-x'\|/(1-\tfrac{\kappa^2}{24} (\tfrac\pi2 \|x-x'\|)^2) \\
&\le \|x-x'\| (1+ c_0 \kappa^2 \|x-x'\|^2),
\end{align*}
where the last inequality holds if $\kappa \|x-x'\|$ is sufficiently small.
In view of that, it suffices to prove that there is $\tau > 0$ such that $\delta_\S(x,x') \le \pi/\kappa$ when $x,x'\in \S$ satisfy $\|x-x'\| \le \tau$.
This is true because \proref{coincide} guarantees that $\delta_\S$ is continuous as a function on the compact set $\S \times \S$.
\end{proof}

\begin{rem}
The quantity $\tau$ can be specified in terms of the reach of $\S$ \citep{federer1959curvature}, which \cite{bernstein2000graph} call the minimum branch separation.
\end{rem}

\begin{rem}
The lower bound in \eqref{approx-S1} is a substantial improvement over \citep[Prop 6.3]{1349695}, which gives $\delta_\S(x,x') - \frac\kappa2 \delta_\S(x,x')^2 \le \|x-x'\|$ when $\|x-x'\| \le 1/2\kappa$.  
\end{rem}

We now have all the ingredients to establish a bound that complements \prpref{unconstrained-ub}.  
Such a bound is already available in the work of \cite{bernstein2000graph} in a somewhat more restricted setting where it is assumed that $\S$ is a compact $C^2$ and geodesically convex submanifold.

\begin{prp} \label{prp:unconstrained-lb}
Suppose $\S \subset \bbR^D$ is compact and satisfies \proref{coincide} and \proref{finite}.
Consider a sample $\cX=\{x_1, \dots, x_N\}\subset \S$.  For $\rad > 0$, form the corresponding $\rad$-ball graph.  Let $c_0$ and $\tau$ be defined per \lemref{approx-S}.  
When $\rad \le \tau$ and $\kappa \rad \le 1/3$, we have
\beq \label{unconstrained-lb}
\delta_\S(x,x') \le (1 + c_0 \rad^2) \Delta_\rad(x, x'), \quad \forall x,x' \in \cX.
\eeq
\end{prp}

\begin{proof}
Fix $x, x' \in \cX$ such that $\Delta_\rad(x, x') < \infty$, for otherwise there is nothing to prove.  Let $x = x_{i_0}, x_{i_1}, \dots, x_{i_m} = x'$ define a shortest path in the graph joining $x$ and $x'$, so that 
$\Delta_\rad(x,x') = \sum_{j=0}^{m-1} \Delta_j$, where $\Delta_j := \|x_{i_j} - x_{i_{j+1}}\|$.
Define $a = \delta_\S(x,x')$ and $a_j = \delta_\S(x_{i_j},x_{i_{j+1}})$ for $j = 0, \dots, m-1$.
Since $\Delta_j \le r \le \tau$, by \lemref{approx-S}, $\Delta_j \min(\frac\pi2, 1+ c_0 \kappa^2 \Delta_j^2) \ge a_j$.  By assumption, $\kappa\rad \le 1/3$, and this is seen to force $1+ c_0 \kappa^2 \rad^2 \le \pi/2$, which then implies that $a_j \le \Delta_j + c_0 \kappa^2 \Delta_j^3$.
We thus have 
\[
a \le \sum_{j=0}^{m-1} a_j 
\le \sum_{j=0}^{m-1} (\Delta_j + c_0 \kappa^2 \Delta_j^3) 
\le \sum_{j=0}^{m-1} \Delta_j (1 + c_0 \kappa^2 \rad^2)
= (1 + c_0 \kappa^2 \rad^2) \Delta_\rad(x,x'). 
\qedhere
\] 
\end{proof} 


\section{Curvature-constrained shortest paths and their approximation} \label{sec:constrained}

In this section, we define the curvature-constrained intrinsic semi-metric on a subset and consider its approximation by a curvature-constrained pseudo-semi-metric based on a neighborhood graph built on a finite sample of points from the surface.
In \secref{constrained-intrinsic-metric} we define the curvature-constrained semi-metric on a given subset, and list a few of its properties.
In \secref{curvature} we define a notion of discrete curvature which has useful consistency properties.
In \secref{constrained-pseudo-metric} we define a new notion of neighborhood graph and a pseudo-semi-metric based on shortest path distances in that graph.
In \secref{constrained-approximation} we show that this pseudo-metric can be used to approximate its continuous counterpart.

\subsection{The curvature-constrained intrinsic semi-metric on a surface}
\label{sec:constrained-intrinsic-metric}
A notion of curvature-constrained semi-metric on a subset $\S$ is obtained from its intrinsic metric defined in \secref{intrinsic-metric} by adding a curvature constraint.  In more detail, for $\kappa > 0$ and $x, x' \in \S$, define 
\beq \label{delta_kappa}
\delta_{\S, \kappa}(x, x') = \inf\big\{a : \text{there is $\gamma$ as in \eqref{delta_S} with curvature bounded by $\kappa$}\big\}.
\eeq
By convention, if there is no path as in \eqref{delta_kappa}, then $\delta_{\S, \kappa}(x, x') = \infty$.  

\begin{rem}
$\delta_{\S, \kappa}(x, x')$ is thus the length of the shortest path on $\S$ joining $x$ and $x'$ among those with curvature bounded pointwise by $\kappa$.  
\end{rem}

Compared with the (unconstrained) intrinsic metric \eqref{delta_S}, we always have, for any subset $\S$, and for and any $\kappa \ge 0$,
\beq
\delta_\S(x,x') \le \delta_{\S,\kappa}(x,x'), \quad \forall x,x' \in \S.
\eeq 

The semi-metric $\delta_{\S, \kappa}$ is typically not a metric, as it may not satisfy the triangle inequality.  

\begin{ex}[No triangle inequality]
Indeed, consider the L-shape curve $\S = \{0\} \times [0,1] \cup [0,1] \times \{0\} \subset \bbR^2$ and fix any $\kappa \ge 0$ finite.  Then $\delta_{\S, \kappa}((0,0),(0,1)) = \delta_{\S, \kappa}((0,0),(1,0)) = 1$ and $\delta_{\S, \kappa}((1,0),(0,1)) = \infty$.  
The same is true even if $\delta_{\S, \kappa}$ is finite.  Indeed, take the figure eight curve $\S = \{(\sin(t), \sin(2t)) : t \in [0,2\pi]\}$.  It self-intersects at the origin and has finite curvature $\kappa_0$.  If we take $\kappa \in [\kappa_0, \infty)$, then 
\[\delta_{\S,\kappa}((-1/\sqrt{2}, 1),(1/\sqrt{2}, 1)) >  \delta_{\S,\kappa}((1/\sqrt{2}, 1),(0, 0)) + \delta_{\S,\kappa}((-1/\sqrt{2}, 1),(0,0)),\] 
since the shortest path joining $(0,0)$ and $(-1/\sqrt{2}, 1)$ and the shortest path joining $(0,0)$ and $(1/\sqrt{2}, 1)$ cannot be concatenated to form a curve with finite curvature everywhere.
\end{ex}

That said, there is an obvious case, important in our context, where $\delta_{\S, \kappa}$ is a true metric.
\begin{lem}
\label{lem:coincide}
When $\S \subset \bbR^D$ satisfies \proref{finite}, for any $\kappa_0 \ge \kappa$, $\delta_{\S, \kappa_0}$ coincides with $\delta_\S$, and in particular is a metric on $\S$.
\end{lem}


The following two lemmas are the equivalent of \lemref{delta_S} for the $\kappa$-curvature-constrained semi-metric.

\begin{lem} \label{lem:curv-lim}
Consider $\S \subset \bbR^D$ compact and a sequence of curves $\gamma_m \subset \S$ with curvature bounded by $\kappa$ and such that $\sup_m \Lambda(\gamma_m) < \infty$.  Then there is a subsequence $(\gamma_{m_k})$ and a curve $\gamma \subset \S$ with curvature bounded by $\kappa$ such that $\gamma_{m_k} \to \gamma$ uniformly and $\Lambda(\gamma_{m_k}) \to \Lambda(\gamma)$ as $k \to \infty$.
\end{lem}

Note that without the assumption of bounded curvature the convergence in length is not guaranteed.  Indeed, take the spiral $\Gamma$ described in \exref{spiral} and define $\gamma_m = \Gamma([m,\infty))$.  In that case, we clearly have $\gamma_m \to \{\origin\}$ in Hausdorff metric, yet $\Lambda(\gamma_m) = \infty$ for all $m$.  

\begin{proof}
Let $a_m = \Lambda(\gamma_m)$ and $a = \sup_m a_m < \infty$.  
Assuming $\gamma_m$ is unit-speed, let $\zeta_m : [0,1] \to \S$ be defined as $\zeta_m(t) = \gamma_m(a_m t)$.  Note that $\zeta_m$ is $a_m$-Lipschitz and $\dot\zeta_m$ is well-defined and Lipschitz with constant $a_m^2 \kappa$.
Since, in particular, $\dot\zeta_m$ is Lipschitz with constant $a^2 \kappa$, by the Arzel\`a-Ascoli theorem there is $M_1 \subset \bbN$ and $\nu : [0,1] \to \bbR^D$ such that $(\dot\zeta_m : m \in M_1)$ converges to $\nu$ uniformly over $[0,1]$. 
Since $\zeta_m(0) \in \S$ and $\S$ is compact, there is $M_2 \subset M_1$ and $x \in \S$ such that $(\zeta_m(0) : m \in M_2)$ converges to $x$.
Define $\zeta(t) = x +  \int_0^t \nu(s) {\rm d}s$ and let $\gamma$ be the curve $\{\zeta(t) : t \in [0,1]\}$.  
We have
\begin{align}
\sup_{t \in [0,1]} \|\zeta_m(t) - \zeta(t)\| 
&\le \|\zeta_m(0) - x\| + \sup_{t \in [0,1]} \|\dot\zeta_m(t) - \nu(t)\| \\
&\to 0, \quad \text{along } m \in M_2.
\end{align}
And since $\S$ is closed, we have $\zeta(t) = \lim_{m \in M_2} \zeta_m(t) \in \S$ for all $t$, so that $\gamma \subset \S$.  Hence, $(\gamma_m : m \in M_2)$ converges uniformly to $\gamma$.
We also have 
\begin{align}
a_m = \Lambda(\gamma_m) = \Lambda(\zeta_m) 
&= \int_0^1 \|\dot\zeta_m(t)\| {\rm d}t \\
&\to \int_0^1 \|\nu(t)\| {\rm d}t = \int_0^1 \|\dot\zeta(t)\| {\rm d}t = \Lambda(\zeta) = \Lambda(\gamma) =: a_\infty,
\end{align}
along $m \in M_2$.  
If $a_\infty = 0$, then $\gamma = \{x\}$ and has curvature 0 (by convention).  If $a_\infty > 0$, then $s \to \zeta(s/a_\infty)$ is a unit-speed parameterization of $\gamma$ (which we also denote by $\gamma$), because, for all $t$, $a_m = \|\dot\zeta_m(t)\| \to \|\nu(t)\|$ and $a_m \to a_\infty$, along $m \in M_2$.
We have that $\dot\gamma$ is $\kappa$-Lipschitz, because $\dot\gamma_m$ is $\kappa$-Lipschitz for all $m$ and $\dot\gamma_m \to \dot\gamma$ pointwise along $m \in M_2$, and thus $\gamma$ has curvature at most $\kappa$.
\end{proof}

\begin{lem} \label{lem:delta_kappa}
Take $x,x' \in \S \subset \bbR^D$.  
If $\S$ is closed and $\delta_{\S,\kappa}(x,x') < \infty$, the infimum in \eqref{delta_kappa} is attained.
\end{lem}

\begin{proof}
Take $x,x' \in \S$ such that $a := \delta_{\S, \kappa}(x,x') < \infty$.
By definition, there is a sequence $(a_m)$ converging to $a$  such that, for each $m$, there is a curve $\gamma_m \subset \S$ of length $a_m$ and of curvature at most $\kappa$ joining $x$ and $x'$.  We then apply \lemref{curv-lim} to get a subsequence $(\gamma_{m_k})$ and a curve $\gamma \subset \S$ with curvature bounded by $\kappa$ such that $\gamma_{m_k} \to \gamma$ uniformly, implying that $\gamma$ joins $x$ and $x'$, as well as $\Lambda(\gamma) = \lim_k \Lambda(\gamma_{m_k}) = a$. 
\end{proof}

\begin{lem} \label{lem:delta_kappa2}
Suppose $\S \subset \bbR^D$ is open and connected.  Then for any $x,x' \in \S$, $\lim_{\kappa \to \infty} \delta_{\S,\kappa}(x,x') = \delta_\S(x,x') < \infty$.  In particular, for any $x,x' \in \S$, there is $\kappa > 0$ such that $\delta_{\S,\kappa}(x,x') < \infty$.
\end{lem} 

\begin{proof}
Fix $x,x' \in \S$.  By \lemref{delta_S}, $a := \delta_\S(x,x') < \infty$.  
For $m \ge 1$ integer, let $a_m = a + 1/m$.  By definition, there is a 1-Lipschitz function $\gamma_m : [0,a_m] \to \S$ such that $\gamma(0) = x$ and $\gamma_m(a_m) = x'$.  
Since $\gamma_m$ (as a curve) is compact and entirely within the open $\S$, by \lemref{compact-in-open} there is $h_m > 0$ such that $\gamma_m \oplus B(0,h_m) \subset \S$.  Consider the polygonal line, denoted $p_m$, assumed to be parameterized by arc-length, joining $\gamma_m(j a_m/n_m)$ for $j = 0, \dots, n_m$ where$n_m = \lceil 2 a_m/h_m \rceil$.
Note that $\Lambda(p_m) \le \Lambda(\gamma_m)$, and also
\begin{align*}
\H(p_m, \gamma_m) 
&\le \sup_t \|p_m(t) - \gamma_m(t)\| \\
&\le \max\Big( \sup_t \min_j \|p_m(t) - \gamma_m(j a_m/n_m)\|, \ \sup_t \min_j \|\gamma_m(t) - p_m(j a_m/n_m)\|\Big) \\
&= \max\Big( \sup_t \min_j \|p_m(t) - p_m(j a_m/n_m)\|, \ \sup_t \min_j \|\gamma_m(t) - \gamma_m(j a_m/n_m)\|\Big) \\
&\le a_m/n_m \le h_m/2.
\end{align*}
We used \eqref{hausdorff} in the first line, the fact that $p_m(j a_m/n_m) = \gamma_m(j a_m/n_m)$ for all $j$ in the third line, and the fact that $p_m$ and $\gamma_m$ are 1-Lipschitz in the fourth line.
We may thus conclude that $p_m \subset \gamma_m \oplus B(0,h_m/2)$. 

We now smooth $p_m$ at the vertices to obtain a function with bounded second derivative almost everywhere.  Consider two consecutive segments of $p_m$.  Zoom in on a ball centered at their intersection point (vertex) and of small enough radius.  After an appropriate change of coordinates and a projection onto the plane defined by the two line segments, in that neighborhood $p_m$ can be made to correspond to the graph of the function $t \to \alpha |t|$ on $(-c, c)$.  This is illustrated in \figref{smooth-polygon}.
We assume the radius of the ball is small enough that $c \le h_m/\alpha$.  
In these coordinates, consider the graph of the function $t \to \frac\alpha{2c} t^2 + \frac{\alpha c}2$, and replace the V-shape piece of $p_m$ in that neighborhood with that piece of parabola.  Do that for all pairs of consecutive line segments of $p_m$ and name the resulting curve $\zeta_m$.  By construction, $\zeta_m$ has bounded curvature, $\Lambda(\zeta_m) \le \Lambda(p_m)$, and $\zeta_m \subset p_m \oplus B(0,h_m/2)$ since $\frac{\alpha c}2 \le h_m/2$, and $\zeta_m$ joins $x$ and $x'$.
In addition, 
\bitem
\item $\Lambda(\zeta_m) \le \Lambda(\gamma_m)$ since $\Lambda(p_m) \le \Lambda(\gamma_m)$;
\item $\zeta_m \subset \gamma_m \oplus B(0,h_m) \subset \S$, since $p_m \oplus B(0,h_m/2) \subset \gamma_m \oplus B(0,h_m)$ by an application of the triangle inequality. \qedhere
\eitem
\end{proof}

\begin{figure}[ht]
\begin{center}
\centerline
{\includegraphics[width=3in]{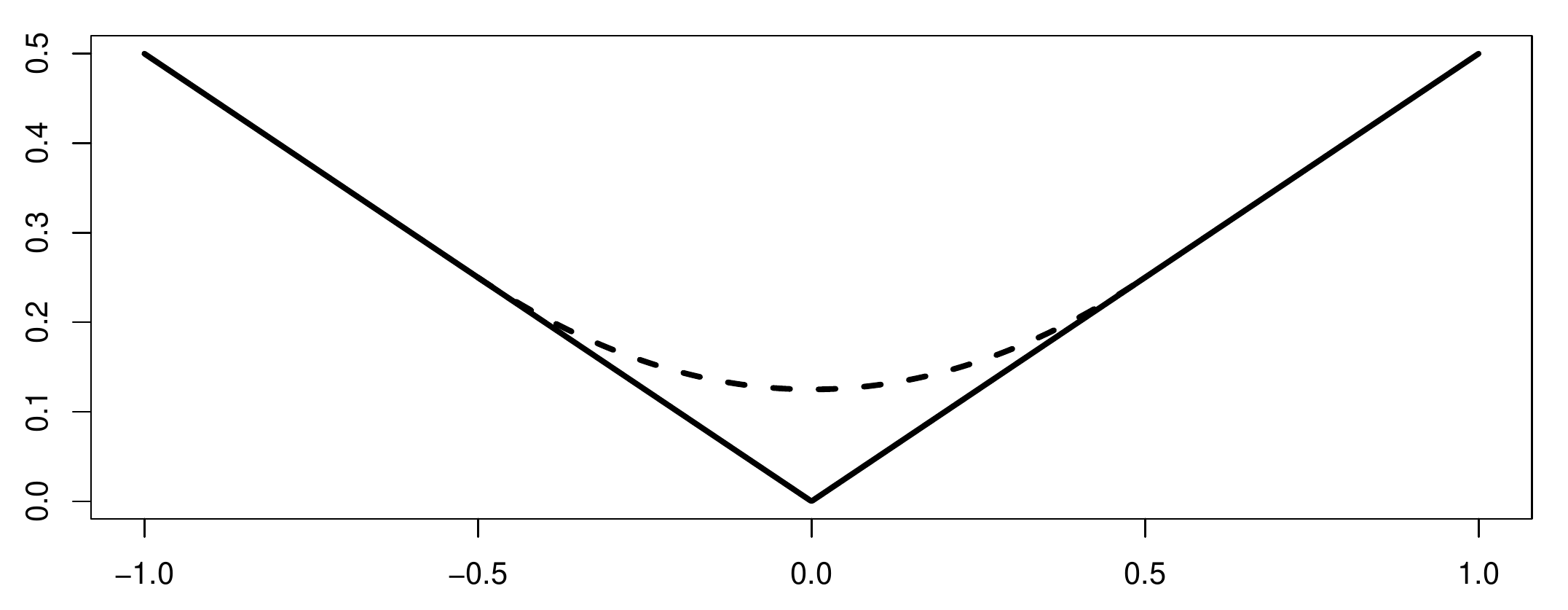}}
\caption{\small Two consecutive line segments of a possibly longer polygonal line $p_m$ and a quadratic approximation (dashed) $\zeta_m$, resulting in shorter length and bounded curvature.}
\label{fig:smooth-polygon}
\end{center}
\vskip -0.2in
\end{figure}

The following results concern a smooth submanifold $\S$, where the shortest paths are known to have uniformly bounded curvature.

\begin{lem} \label{lem:finite0}
Assume that $\S \subset \bbR^D$ is a compact and connected $C^2$ submanifold without boundary.  Let $\kappa_\S$ denote the maximum (unsigned) principal curvature at any point on $\S$.  Then the shortest paths on $\S$ have curvature at most $\kappa_\S$.
\end{lem}

\begin{proof}
The result is due to the Hopf-Rinow theorem \citep[Th 6.13]{lee2006riemannian}, which in the present case implies that the shortest paths on $\S$ are geodesics\footnote{ For a definition of geodesics, see Chapter 4 in \citep{lee2006riemannian}}.  
This is combined with the fact that a geodesic $\gamma$ on $\S$ (assumed parameterized by arc-length) has curvature at $x = \gamma(t)$ equal to the second fundamental form of $\S$ at $x$ applied to $(\dot\gamma(t), \dot\gamma(t))$ \citep[Lem 8.5]{lee2006riemannian}.
(Note that $\kappa_\S$ is finite by the fact that $\S$ is compact and $C^2$.)
\end{proof}

The last result generalizes to submanifolds with boundary, although the situation is more complicated in general.\footnote{ For example, the points where a shortest path switches from the (relative) interior and the boundary can have a closure of positive measure; this is true even in the case of a domain \citep{albrecht1991geodesics}.} 
\begin{lem} \label{lem:finite}
Assume that $\S \subset \bbR^D$ is a compact and connected $C^2$ submanifold with boundary $\partial S$ that is also a $C^2$ submanifold.  Define $\kappa_\S$ and $\kappa_{\partial \S}$ as in \lemref{finite0}.
Then the shortest paths on $\S$ have curvature at most $\max(\kappa_\S, \kappa_{\partial \S})$.
\end{lem}

\begin{proof}
Consider a shortest path $\gamma$ between $x,x' \in \S$, assumed to be unit-speed.  It must be the concatenation of shortest paths $\S \setminus \partial \S$ and shortest paths in $\partial \S$.
We learn from \citep[Sec 2]{alexander1987riemannian} that $\gamma$ is twice differentiable except at switch points (where $\gamma$ switches between $\S$ and $\partial\S$), and in particular it has curvature bounded by $\max(\kappa_\S, \kappa_{\partial \S})$ except at switch points.  We also learn that $\gamma$ must have zero curvature at any accumulation point of switch points.  Thus the only concern might be the isolated switch points.  However, we learn from \citep{alexander1981geodesics} that $\gamma$ must be at least $C^1$.  Therefore, overall, it must have curvature at most $\max(\kappa_\S, \kappa_{\partial \S})$ everywhere.
\end{proof}

The fact that $\S$ is a submanifold is, in fact, not necessary for points to be joined by curvature constrained paths, as the following extension establishes.  In particular, self-intersections are possible.

\begin{lem} \label{lem:finite-extend}
Assume that $\cU \subset \bbR^d$ is compact and such that $\delta_{\cU, \kappa_0}(u, u') < \infty$ for all  $u, u' \in \cU$ for some $\kappa_0 \ge 0$.  Let $\psi : \cU \to \S \subset \bbR^D$ be twice differentiable, surjective, and such that $D_u \psi$ is nonsingular for all $u \in \cU$.  Then there is $\kappa$ depending only on $\psi$ and $\kappa_0$ such that $\delta_{\S, \kappa}(x,x') < \infty$ for all $x,x' \in \S$.\end{lem}

\begin{proof}
Take $x, x' \in \S$ and let $u,u' \in \cU$ be such that $\psi(u) = x$ and $\psi(u') = x'$.  Let $\alpha : [0,a] \to \cU$ be a unit-speed path on $\cU$ with curvature at most $\kappa_0$ joining $u$ and $u'$.  The existence of such an $\alpha$ relies on the fact that $\delta_{\cU, \kappa_0}(u, u') < \infty$.  Define $\zeta(t) = \psi(\alpha(t))$, so that $\zeta: [0,a] \to \S$ joins $x$ and $x'$.  Note that $\zeta$ is $C^2$ with 
\begin{align*}
\dot\zeta(t) &= D_{\alpha(t)} \psi (\dot\alpha(t)), \\ 
 \ddot\zeta(t) &= D^2_{\alpha(t)} \psi(\dot\alpha(t),\dot\alpha(t)) + D_{\alpha(t)} \psi (\ddot\alpha(t)),
\end{align*}
where $D_u^2 \psi$ denotes the differential of $\psi$ of order 2 at $u$.
By the fact that $D_u \psi$ is nonsingular and continuous in $u$, its smallest singular value, minimized over $u \in \cU$ is strictly positive.  If we denote this by $a_1$, we have that $a_1 >0$ and that $\|D_u \psi (v)\| \ge a_1 \|v\|$ for all $u$ and $v$.  
Similarly, there are reals $b_1$ and $b_2$ such that $\|D_u \psi (v)\| \le b_1 \|v\|$ and $\|D^2_u \psi (v,v)\| \le b_2 \|v\|^2$ for all $u$ and $v$.  Therefore,
\begin{align*}
\|\dot\zeta(t)\| &\ge a_1 \|\dot\alpha(t)\| = a_1, \\
\|\ddot\zeta(t)\| &\le \|D^2_{\alpha(t)} \psi\| \|\dot\alpha(t)\|^2 + \|D_{\alpha(t)} \psi\| \|\ddot\alpha(t)\| \le b_2 + b_1 \kappa_0.
\end{align*}
Thus, we have
\[
\curv(\zeta,t) = \frac{\|\dot\zeta(t) \wedge \ddot\zeta(t)\|}{\|\dot\zeta(t)\|^3} \le \frac{\|\ddot\zeta(t)\|}{\|\dot\zeta(t)\|^2} 
\le \frac{b_2 + b_1 \kappa_0}{a_1}.
\qedhere
\]
\end{proof}

\subsection{A notion of curvature for polygonal lines}
\label{sec:curvature}
A polygonal line has infinite curvature at any vertex (excluding the endpoints if the line is open).  Nevertheless, it is possible to define a different notion of curvature specifically designed for polygonal lines that will prove useful later on when we approximate smooth curves with polygonal lines.

For an ordered triplet of points $(x, y, z)$ in $\bbR^D$, define its angle, denoted $\angle(x,y,z)$, as the following angle $\angle(x-y, z-y) \in [0, \pi]$, and define its curvature as
\beq \label{curv-def}
\curv(x,y,z) = \begin{cases}
1/R(x,y,z), & \text{if } \angle(x,y,z) \ge \frac\pi2, \\
\infty, & \text{otherwise},
\end{cases}
\eeq
where $R(x,y,z)$ denote the radius of the circle passing through $x,y,z$  ---  with $R(x,y,z) = \infty$ if $x,y,z$ are aligned.  This is illustrated in \figref{curvature}.

\begin{figure}[ht]
\begin{center}
\centerline
{\includegraphics[width=3in]{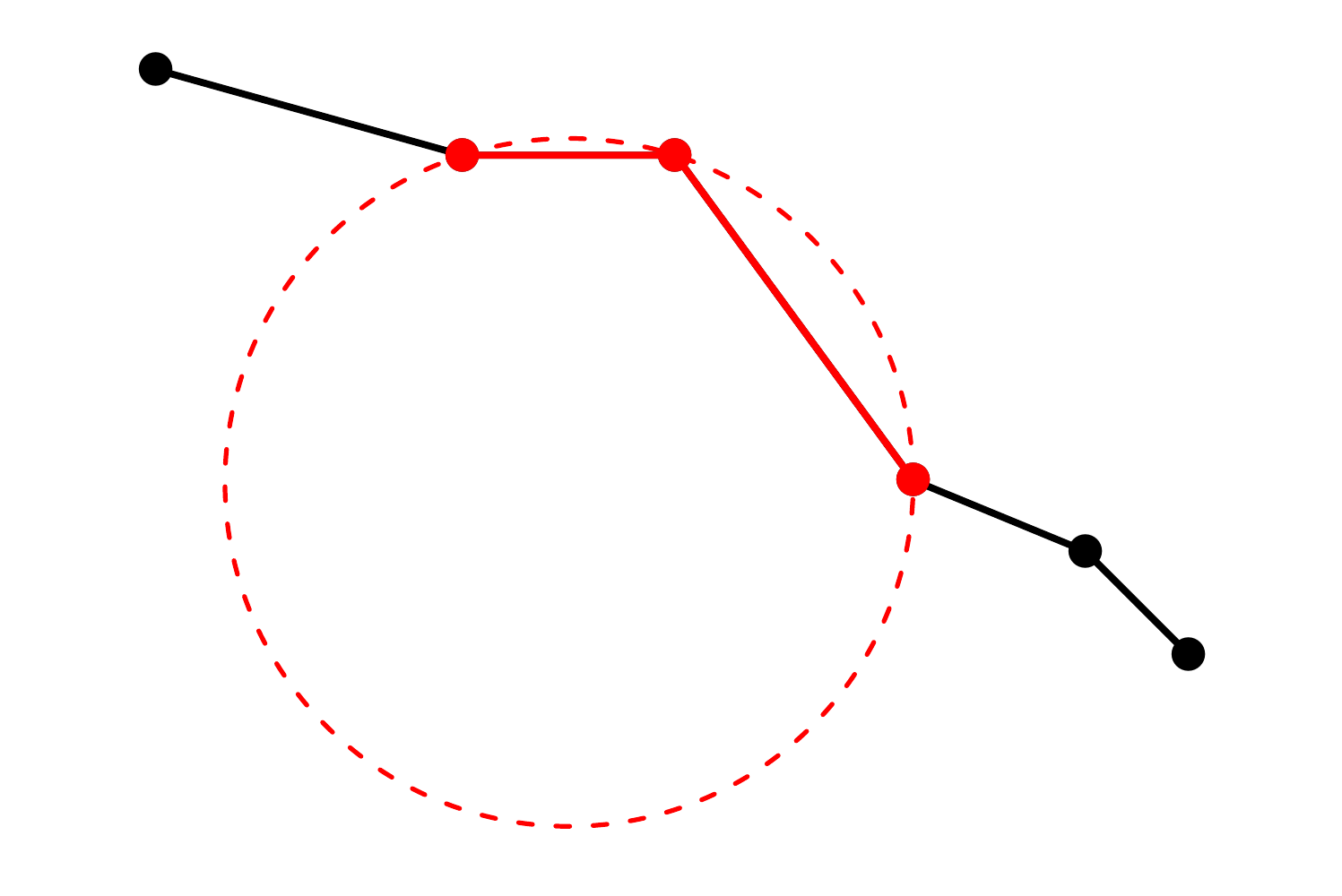}}
\caption{\small The curvature at the middle vertex is defined as the inverse of the radius of the circle (circumradius) passing through this point and the two adjacent vertices on the polygonal line (highlighted in red).}
\label{fig:curvature}
\end{center}
\vskip -0.2in
\end{figure}

Using a well-known expression for the circumradius, we obtain the following.

\begin{lem} \label{lem:curv}
For any distinct points $x,y,z \in \bbR^D$ such that $\angle(x, y, z) \ge \frac\pi2$,
\beq \label{curv-ang}
\curv(x,y,z) = \frac{2 \|(x-y) \wedge (y-z)\|}{\|x-y\| \|y-z\| \|z-x\|} = \frac{2\sin\angle(x, y, z)}{\|x-z\|}.
\eeq
\end{lem} 

Other notions of discrete curvature exist in the literature, as discussed in \citep[Sec.~2.2]{MR2934178}.  We chose to work with this particular one because of the following consistency property.
Recall the definition \eqref{curv-gamma}.

\begin{lem} \label{lem:curvature}
Consider a curve $\gamma : (a,b) \to \bbR^D$ which is twice continuously differentiable.  Then, holding $s \in (a,b)$ fixed, we have\footnote{ Here, $r \nearrow s$ means that $r$ approaches $s$ from the left, and similarly, $t \searrow s$ means that $t$ approaches $s$ from the right, on the real line.}
\[
\curv(\gamma(r), \gamma(s), \gamma(t)) \to \curv(\gamma, s), \quad \text{as } r \nearrow s \text{ and } t \searrow s.
\]
\end{lem}

Other notions of curvature do not always enjoy this consistency.  For example, those based on angle defect or Steiner's formula \citep{MR2610471} are of the form $f(\angle(x,y,z))$, for some continuous function $f$, and therefore are not consistent in general, since $\angle(\gamma(r), \gamma(s), \gamma(t)) \to \pi$ when $\gamma$ is differentiable at $s$ and $r \nearrow s$ and $t \searrow s$.

\begin{proof}
As usual, we assume that $\gamma$ has been parameterized by arc-length, and because $\gamma$ is assumed twice differentiable, we have $\curv(\gamma,s) = \| {\ddot\gamma(s)}\|$.
We expand $\gamma$ around $s$, to get
\begin{align*}
\gamma(r) - \gamma(s) &= (r-s) \dot\gamma(s) + \tfrac12 (r-s)^2 \ddot\gamma(s) + o(r-s)^2, \\
\gamma(t) - \gamma(s) &= (t-s) \dot\gamma(s) + \tfrac12 (t-s)^2 \ddot\gamma(s) + o(t-s)^2, \\
\gamma(t) - \gamma(r) &= (t-r) \dot\gamma(s) + \tfrac12 [(t-s)^2 - (r-s)^2] \ddot\gamma(s) + o(t-s)^2 + o(r-s)^2.
\end{align*}
This implies that 
\begin{align*}
\gamma(r) - \gamma(s) &= (1+o(1)) (r-s) \dot\gamma(s), \\
\gamma(t) - \gamma(s) &= (1+o(1)) (t-s) \dot\gamma(s), \\
\gamma(t) - \gamma(r) &= (1+o(1)) (t-r) \dot\gamma(s),
\end{align*}
where in the last line we used the fact that $o(t-s)^2 + o(r-s)^2 = o(t-r)^2$ since $r < s < t$.  
We thus obtain
\[
\cos \angle(\gamma(r), \gamma(s), \gamma(t)) = \frac{\<\gamma(s) - \gamma(r), \gamma(s) - \gamma(t)\>}{\|\gamma(s) - \gamma(r)\| \|\gamma(s) - \gamma(t) \|} = 1 + o(1),
\]
by the fact that $\|\dot\gamma(s)\| = 1$.  Thus, $\angle(\gamma(r), \gamma(s), \gamma(t)) \ge \pi/2$ eventually.
Assuming this is the case, applying \lemref{curv}, we have
\beq\label{curv-gamma-calculations}
\curv(\gamma(r), \gamma(s), \gamma(t))=\frac{2\left \| (\gamma(r)-\gamma(s))\wedge (\gamma(t)-\gamma(s))\right \|}{\left \| \gamma(r)-\gamma(s) \right \|\left \| \gamma(t)-\gamma(s) \right \|\left \|\gamma(t)-\gamma(r)\right \|}.
\eeq
From the same derivations, we also obtain
\beq \label{3norms}
\left \| \gamma(r)-\gamma(s) \right \|\left \| \gamma(t)-\gamma(s) \right \|\left \|\gamma(t)-\gamma(r)\right \| = (1+o(1)) (s-r) (t-s) (t-r).
\eeq
Recalling that $\|u \wedge v\| = \|u\| \|v\| \sin \angle(u,v)$ for any vectors $u,v \in \bbR^d$, we also have
\begin{align*}
(\gamma(r)-\gamma(s))\wedge (\gamma(t)-\gamma(s)) 
&= \tfrac12 (r-s)^2 (t-s) \ddot \gamma(s) \wedge \dot\gamma(s) + \tfrac12 (r-s) (t-s)^2 \dot \gamma(s) \wedge \ddot\gamma(s) \\
& \qquad + o((r-s)^2(t-s) + (r-s)(t-s)^2) \\
&= \tfrac12 (s-r) (t-s) (t-r) \dot \gamma(s) \wedge \ddot\gamma(s) + o((s-r) (t-s) (t-r)),
\end{align*}
again using the fact that $r < s < t$.  This implies that 
\beq \label{crossnorm}
\|(\gamma(r)-\gamma(s))\wedge (\gamma(t)-\gamma(s))\| = \tfrac12 (s-r) (t-s) (t-r) \|\ddot\gamma(s)\| + o((s-r) (t-s) (t-r)),
\eeq
using the fact that $\dot\gamma(s)$ and $\ddot\gamma(s)$ are orthogonal.  

We conclude the proof by plugging in \eqref{3norms} and \eqref{crossnorm} in \eqref{curv-gamma-calculations}, and simplifying.
\end{proof} 

While \lemref{curvature} is qualitative in nature, we will also need a quantitative bound.  The following result provides such a bound.  The proof is more delicate.

\begin{lem}\label{lem:shortest-arc}
Let $\gamma$ be a simple curve with curvature bounded above by $\kappa$.  Then $\curv(x,y,z)\leq \kappa$ for all $x, y, z \in\gamma$ distinct such that $y$ is between $x$ and $z$ on $\gamma$ and $\|x - z\| \le 2/\kappa$.
\end{lem}

\begin{proof}
Take $\gamma$ and $x, z$ as in the statement.
By continuity, it is enough to prove the result when $\|x - z\| < 2/\kappa$.
WLOG, we assume that $\gamma$ has endpoints $x$ and $z$, and that $\gamma$ is parameterized by arc length and let $\ell$ denote its length, so that $\gamma : [0,\ell] \to \bbR^D$.
In that case, \lemref{approx} implies that $\ell < \pi/\kappa$.

Define $\bbB(x,z,\kappa)$ as the set of all open balls $B$ of radius $1/\kappa$ such that $x,z\in\partial B$.  
The condition $\|x - z\| < 2/\kappa$ guarantees that $\bbB(x,z,\kappa)$ is not empty.
Define $\cV_\kappa$ as the intersection of all (closed) balls belonging to $\bbB(x,z,\kappa)$, that is,
\[
\cV_\kappa=\bigcap_{B\in\bbB(x,z,\kappa)} \bar B.
\]

\medskip\noindent {\em Claim 1: $y\in\cV_\kappa$ for all $y\in\gamma$.}
Suppose that there exists $y\in\gamma$ such that $y\notin\cV_\kappa$.
Then there exists $B\in\bbB(x,z,\kappa)$ such that $y\notin \bar B$.  Let $b$ denote the center of $B$, and recall that $B$ has radius $1/\kappa$.
By continuity, there exists $0 \le s < t \le \ell$ such that 
\beq\label{shortestarc0}
\gamma(s) \in \partial B, 
\quad \gamma(t) \in \partial B, \quad 
\gamma((s,t)) \subset \bar B^\comp.
\eeq
Denote $\zeta$ the shortest path on $\partial B$ joining $\gamma(s)$ and $\gamma(t)$, which is indeed uniquely defined since $B$ has radius $1/\kappa$ and $\|\gamma(s) - \gamma(t)\| \le \ell < \pi/\kappa$ as we saw above.
Also, denoted $\gamma_* = \gamma([s,t])$.

First, we claim that $\zeta$ is not longer than $\gamma_*$, meaning that $\Lambda(\zeta) \le \Lambda(\gamma_*)$.
To see this, let $\zeta_* = P \gamma_*$, where $P$ here denotes the metric projection of $\gamma_*$ onto $B$.  Since $P$ is 1-Lipschitz, we have $\Lambda(\zeta_*) \le \Lambda(\gamma_*)$ by \lemref{Lip-length}.
And since $\zeta_*$ is a path on $\partial B$ joining $\gamma(s)$ and $\gamma(t)$, and $\zeta$ is the unique shortest such path, $\Lambda(\zeta_*) > \Lambda(\zeta)$, unless $\zeta_* = \zeta$.
Hence, we indeed have that $\Lambda(\zeta) \le \Lambda(\gamma_*)$.

Next, we reverse this relationship.  Indeed, we apply \lemref{approx} together with the fact that $t-s \le \ell \le \pi/\kappa$, and then use the fact that $\zeta$ is a piece of circle of radius $1/\kappa$ joining $\gamma(s)$ and $\gamma(t)$, to get
\[
\frac{2}{\kappa}\sin\left(\frac{\kappa}{2}\Lambda(\gamma_*)\right)
= \frac{2}{\kappa}\sin\left(\frac{\kappa}{2}(t-s)\right)
\le \|\gamma(t)-\gamma(s)\| 
= \frac{2}{\kappa}\sin\left(\frac{\kappa}{2}\Lambda(\zeta)\right).
\]
Using the fact that the sine function is increasing on $[0,\frac{\pi}{2}]$, and again using the fact that $t-s \le \ell \le \pi/\kappa$, this implies that $\Lambda(\gamma_*) \le \Lambda(\zeta)$.  

We can therefore conclude that $\Lambda(\gamma_*) = \Lambda(\zeta)$, which then implies that $\Lambda(\gamma_*) = \Lambda(\zeta_*)$.  
However, by \lemref{projection-ball}, this is only possible if $\gamma_*$ coincides with $\zeta_*$, which is in contradiction with the fact that $\gamma_*$ only intersects $\partial B$ at its endpoints.

\medskip
\noindent {\em Claim 2: $\curv(x,y,z) \le 1/\kappa$ for all $y \in \cV_\kappa$.}  Take $y \in \cV_\kappa$ and consider the affine plane generated by $x,y,z$.  We work in that plane hereafter.  There exists two distinct points $b$ and $b'$ such that $y\in \bar B(b, 1/\kappa)\cap \bar B(b', 1/\kappa)$ and $x,z\in\partial B(b,1/\kappa)\cap \partial B(b',1/\kappa)$.  Assume WLOG that $y$ and $b$ are on different sides of the line $(xz)$.  Let $L$ denote the line passing through $y$ and perpendicular to $(xz)$, and let $w$ be the point at the intersection of $L$ and $\partial B(b,1/\kappa)$.  Then $\angle (x,y,z) \ge \angle (x,w,z)$, and because $w$ is on the (short) arc defined by $x$ and $z$ on the circle $\partial B(b,1/\kappa)$, $\angle (x,w,z) \ge \pi/2$.
Hence, $\angle (x,y,z) \ge \pi/2$.
Moreover, by \lemref{curv}, 
\[
\curv(x,y,z)
= \frac{2\sin\angle(x, y, z)}{\|x-z\|}
\le \frac{2\sin\angle(x, w, z)}{\|x-z\|}
= \curv(x,w,z) 
= \kappa,
\] 
where the last equality is by definition of the curvature.
\end{proof}

\subsection{A neighborhood graph and its curvature-constrained semi-metric}
\label{sec:constrained-pseudo-metric}
We now define a curvature-constrained analog of the metric defined in \secref{pseudo-metric}.  
Recall the $r$-ball neighborhood graph defined in \secref{pseudo-metric}, also based on a sample $\cX = \{x_1, \dots, x_N\} \subset \S$.
For $\kappa > 0$, define\footnote{ The computation of curvature-constrained shortest path distances can be done by adapting Dijkstra's algorithm.  It is implemented in Algorithm~1 of \citep{babaeian2015nonlinear}.}  
\begin{gather}
\Lambda^*_{r,\kappa}(i,j) = \min\big\{\Lambda_r(k_1, \dots, k_m): m \ge 1, k_1 = i, k_m = j, \notag \\ 
\qquad \qquad \qquad \qquad \qquad \qquad \qquad \qquad \qquad 
{\textstyle\max_l} \curv(x_{k_{l-1}}, x_{k_l}, x_{k_{l+1}}) \le \kappa\big\}. \label{Lambda_kappa}
\end{gather}
Equivalently, this is the length of the shortest polygonal line with curvature bounded by $\kappa$ joining $x_i$ and $x_j$ in the graph.
(If no such path exists, it is equal to infinity by convention.)
Note that it is only a semi-metric on the graph in general.
For two sample points, $x_i, x_j \in \cX$, define 
\beq\label{Delta_kappa}
\Delta_{\rad, \kappa}(x_i, x_j) = \Lambda^*_\rad(i,j),
\eeq
thus defining a semi-metric on the sample $\cX$.
As we did in \remref{extension}, this can be extended to a pseudo-semi-metric on the surface $\S$.

For technical reasons, we will also work with a different, uncommon kind of neighborhood graph:
\bitem \setlength{\itemsep}{0in}
\item {\em $(\rad,\alpha)$-annulus graph:} $i \sim j$ if and only if  $\alpha \rad \le \|x_i - x_j\| \le \rad$,
\eitem
which yields a weighted graph on $\{1, \dots, N\}$ with weights denoted $w_{r,\alpha}(i,j)$ and defined analogously \eqref{weight}, except that the neighborhood structure is different.  
Let $\Lambda^*_{r,\alpha,\kappa}(i,j)$ denote the corresponding shortest path distance, and based on that, define $\Delta_{r,\alpha,\kappa}(x_i,x_j) = \Lambda^*_{r,\alpha,\kappa}(i,j)$ for all $i, j \in [N]$.

Note that, for any $\kappa > 0$ and any $\alpha \in [0,1)$,
\beq\label{Delta-Delta}
\Delta_r(x,x')
\le \Delta_{r,\kappa}(x,x') 
\le \Delta_{r,\alpha,\kappa}(x,x'), \quad \forall x, x' \in \cX.
\eeq

\begin{rem}
An $(r,\alpha)$-annulus graph may be seen as a regularized $r$-ball graph where the shorter edges have been removed to effectively limit the dynamic range of the edge lengths to $1/\alpha$.  Although we introduce this regularization here to enable our statement of \thmref{constrained-lb}, this sort of regularization may also be useful at an algorithmic level as it sparsifies the neighborhood graph. 
\end{rem}

\subsection{Approximation}
\label{sec:constrained-approximation}
We now consider approximating the curvature-constrained intrinsic semi-metric with the pseudo-semi metric defined on a ball or annulus neighborhood graph.

We first obtain a bound comparable to that satisfied by unconstrained shortest paths in \prpref{unconstrained-ub}, as long as the constraint on the curvature is slightly looser.

\begin{prp} \label{prp:constrained-ub}
Consider $\S \subset \bbR^D$ compact and a sample $\cX=\{x_1, \dots, x_N\}\subset \S$, and let $\eps = \H(\S \mid \cX)$.  For $\rad > 0$ and $\alpha \le 1/4$, form the corresponding $(\rad,\alpha)$-annulus graph.  
There is a numerical constant $C \ge 1$ such that, when $\max(\eps/\rad, \kappa\rad) \le 1/C$ and $\kappa' \ge \kappa + C(\kappa^2 \rad + \eps/\rad^2)$, we have
\[
\Delta_{\rad,\alpha,\kappa'}(x, x') \le (1 + 6 \eps/\rad) \delta_{\S,\kappa}(x,x'), \quad \forall x,x' \in \cX.
\]
\end{prp}

We note that, in view of \eqref{Delta-Delta}, the bound also applies if one works with the $\rad$-ball graph instead.  
Note that the bound is useful when $\rad$ and $\eps/\rad^2$ are both small.

\begin{proof}
As in the proof of \prpref{unconstrained-ub}, we may focus on the case where $\|x - x'\| > r$.
Assume that $a := \delta_{\S, \kappa}(x,x') < \infty$ (for otherwise the bound holds trivially) and let $\gamma: [0,a] \to \S$ be parameterized by arc length and with curvature bounded by $\kappa$, and such that $\gamma(0) = x$ and $\gamma(a) = x'$, which exists by \lemref{delta_kappa}.  
Let $y_j = \gamma(j a/m)$ for $j = 0, \dots, m$, where $m := \lfloor 3 a/\rad \rfloor \ge 3$.  
We will use the fact that, since $a \ge \|x - x'\| > r$, 
\[
r/3 \le a/m \le a/(3a/r -1) < r/2.
\]
Let $x_{i_j}$ be closest to $y_j$ among the sample points.  In particular, $\max_j \|x_{i_j} - y_j\| \le \eps$ by definition of $\eps$.
Proceeding exactly as in the proof of \prpref{unconstrained-ub}, we find that $(x_{i_0}, \dots, x_{i_m})$ forms a path in the $\rad$-ball graph with length bounded from above by $(1 + 6 \eps/\rad) a$.
We now argue that: 
1) this is also a path in the $(\alpha,\rad)$-annulus graph; and
2) its curvature is at most $\kappa + C_1 (\kappa^2 \rad + \eps/\rad^2)$ for some constant $C_1 >0$ depending on $C$.  
Below, we let $A$ denote a positive constant that may change (increase) with each appearance, and may depend on $C$ but not on $\kappa,\rad,\eps$.

The triangle inequality gives $\|x_{i_j} - x_{i_{j+1}}\| \ge \|y_j - y_{j+1}\| - 2 \eps$, and a Taylor development of $\dot\gamma$ of order 1 gives, for $j = 0, \dots, m-1$,
\begin{align*}
\|y_j - y_{j+1}\| 
&= \|\gamma(j a/m) - \gamma((j+1)a/m)\| \\
&= \Big\|\int_{ja/m}^{(j+1)a/m} \dot\gamma(t) {\rm d} t\Big\| \\
&\ge \|(a/m) \dot\gamma((j+1/2)a/m)\| - \tfrac18 (a/m)^2 \|\ddot\gamma\|_\infty 
\ge \rad/3 - A\kappa \rad^2,
\end{align*}
so that 
\[
\|x_{i_j} - x_{i_{j+1}}\| 
\ge \rad/3 - A \kappa \rad^2 - 2 \eps
\ge \rad/4 \ge \alpha \rad,
\]
when $C$ is large enough.  This proves that, indeed, $(x_{i_0}, \dots, x_{i_m})$ forms a path in the $(\rad,\alpha)$-annulus graph.

Next, in the same way, we have $\|y_{j-1} - y_{j+1}\| \ge 2\rad/3 - A \kappa \rad^2$ for $j = 1, \dots, m-1$, which leads to $\|x_{i_{j-1}} - x_{i_{j+1}}\| \ge 2\rad/3 - A \kappa \rad^2 - 2 \eps$.
Hence,
\begin{align*}
\|x_{i_{j-1}} - x_{i_j}\| \|x_{i_j} - x_{i_{j+1}}\| \|x_{i_{j+1}} - x_{i_{j-1}}\| 
&\ge (\rad/3 - A \kappa \rad^2 - 2 \eps)^2 (2\rad/3 - A \kappa \rad^2 - 2 \eps) \\
&\ge (2\rad^3/27)(1 - A (\kappa\rad + \eps/\rad)).
\end{align*}
Next, for all $j$, letting $z_j = x_{i_j} - y_j$, we have
\begin{align*}
\|(x_{i_{j-1}} - x_{i_j}) \wedge (x_{i_j} - x_{i_{j+1}})\| 
&\le \|(y_{j-1} - y_{j}) \wedge (y_{j} - y_{j+1})\| 
+ \|y_{j-1} - y_{j}\| \|z_{j} - z_{j+1}\|
\\&\quad + \|y_{j} - y_{j+1}\| \|z_{j-1} - z_{j}\| 
+ \|z_{j-1} - z_{j}\| \|z_{j} - z_{j+1}\| 
\\&\le \|(y_{j-1} - y_{j}) \wedge (y_{j} - y_{j+1})\| 
+ 2 \rad\eps + (2 \eps)^2,
\end{align*}
using the fact that $\|y_{j-1} - y_{j}\| \le a/m < r/2$ and $\|z_{j} - z_{j+1}\| \le \|z_{j}\| + \|z_{j+1}\| \le 2\eps$.
A Taylor development gives, for $j = 1, \dots, m-1$,
\begin{align*}
y_{j-1} - y_{j} &= \gamma((j-1)\rad/3) - \gamma(j\rad/3) = -(\rad/3) \dot\gamma(j\rad/3) + R_j, \\
y_{j} - y_{j+1} &= \gamma(j\rad/3) - \gamma((j+1)\rad/3) = -(\rad/3) \dot\gamma(j\rad/3) + R'_j,
\end{align*}
where $\max(\|R_j\|, \|R'_j\|) \le \kappa \rad^2/18$.  
With this, we get
\begin{align*}
\|(y_{j-1} - y_{j}) \wedge (y_{j} - y_{j+1})\| 
&\le \|(\rad/3) \dot\gamma(j\rad/3)\| (\|R_j\| + \|R'_j\|) + \|R_j\| \|R'_j\| \\
&\le (\rad/3) (2 \kappa \rad^2/18) + (\kappa \rad^2/18)^2 \\
&\le \kappa \rad^3/27 + A (\kappa \rad^2)^2.
\end{align*}
We thus get, 
\begin{align*}
\curv(x_{i_{j-1}}, x_{i_j}, x_{i_{j+1}}) 
&= \frac{2 \|(x_{i_{j-1}} - x_{i_j}) \wedge (x_{i_{j+1}} - x_{i_j})\|}{\|x_{i_{j-1}} - x_{i_j}\| \|x_{i_{j+1}} - x_{i_j}\| \|x_{i_{j-1}} - x_{i_{j+1}}\|} \\
&\le \frac{2 \big[\kappa \rad^3/27 + A (\kappa \rad^2)^2 + 2 \rad\eps + (2 \eps)^2\big]}{(2\rad^3/27)(1 - A (\kappa\rad + \eps/\rad))} \\
&\le \kappa + A(\kappa^2 \rad + \eps/\rad^2). \qedhere
\end{align*}
\end{proof} 

We obtain below a bound that complements that of \prpref{constrained-ub}.  To better understand what kind of result would be particularly pertinent, suppose that $\S \subset \bbR^D$ is compact and satisfies \proref{finite}.  Under the conditions of \prpref{unconstrained-lb}, we have
\[
\delta_{\S,\kappa}(x, x')
= \delta_\S(x, x')
\le (1 + c_0 \rad^2) \Delta_\rad(x, x')
\le (1 + c_0 \rad^2) \Delta_{\rad, \alpha, \kappa'}(x, x'), \quad \forall x, x' \in \cX,
\]
and so for any $\kappa' > 0$.  We used \lemref{coincide} together with \eqref{Delta-Delta}.  However, such a bound is only useful if $\Delta_{\rad, \alpha, \kappa'}(x, x') < \infty$.
So that the central question is what values of $\kappa'$ make this true for most, if not all, pairs of sample points.
We answer this question in a strong sense by proving that the unconstrained shortest paths (in the annulus graph) satisfy a $\kappa'$-curvature constraint with $\kappa'$ close to $\kappa$.

\begin{thm}\label{thm:constrained-lb}
Suppose $\S \subset \bbR^D$ is compact and satisfies \proref{finite}.
Consider $\cX=\{x_1, \dots, x_N\}\subset \S$, and let $\eps = \H(\S \mid \cX)$.  For $\rad > 0$ and $\alpha \le 1/4$ form the corresponding $(\rad,\alpha)$-annulus graph.
There is a universal constant $C \ge 1$ such that, if $\max\big\{\kappa r, \eps/(\alpha \kappa r^2)\big\} \le 1/C$, the unconstrained shortest paths in the graph have curvature bounded above by $\kappa' := \kappa (1+C \eps/\alpha\kappa^2\rad^3)$.
\end{thm}

Note that the bound is useful when $\rad$ and $\eps/\rad^3$ are both small, and compare with the requirement for \prpref{constrained-ub}.

\begin{proof}[Proof of \thmref{constrained-lb}]
We first note that it suffices to consider a shortest path in the graph with only three vertices, denoted $(x_1, x_2, x_3)$ henceforth WLOG.
Necessarily,
\begin{gather}
\min(\|x_1 - x_2\|, \|x_3 - x_2\|) \ge \alpha r, \label{123-min}\\
\max(\|x_1 - x_2\|, \|x_3 - x_2\|) \le r < \|x_1-x_3\| \le 2r. \label{123-max}
\end{gather}
Assume WLOG that $\|x_1 - x_2\| \le \|x_2 - x_3\|$.
For a point $y$, define $a[y] = \|x_1 - y\|$ and $b[y] = \|x_3 - y\|$, and $\k[y] = \curv(x_1,y,x_3)$, the latter possibly infinite.
We let $\cc = \|x_1 - x_3\|$ and $\theta = \angle(x_1,x_2,x_3)$.
Our goal, therefore, is to bound $\k[x_2]$ from above.
Below, $A_0, A_1, \dots$ denote universal constants greater than or equal to 1.

\medskip\noindent{\em Case 1: Assume that $\cc \ge 2r -6\eps$.}  
For this particular case, let $a$ and $b$ be short for $a[x_2]$ and $b[x_2]$, respectively.  
We have $\min(\aa, \bb) \ge r - 6\eps$. 
For $\eps/r$ small enough, this forces $\min(\aa, \bb) \ge r/2$ and, since $a^2+b^2<c^2$, also $\theta[x_2] \ge \pi/2$.  
Using \eqref{law-cosines}, for example, we have that
\begin{align*}
\k[x_2] 
&= \frac{\sqrt{(a+b+c)(-a+b+c)(a-b+c)(a+b-c)}}{abc} \\
&\le \frac{\sqrt{(4r)(3r)(3r)(6\eps)}}{(r/2)(r/2)(2r-6\eps)} 
\le \sqrt{A_0 \eps/r^3} = \sqrt{A_0 \eps/(\kappa^2 r^3)} \, \kappa\le \kappa\left(1+ A_0 \eps/\kappa^2 r^3\right).
\end{align*}

\medskip\noindent{\em Case~2: Assume that $\cc \le 2r -6\eps$.}  
This implies that $a[x_2] \le r-3\eps$, since we assumed that $a[x_2] \le b[x_2]$.
Let $\gamma$ be a shortest path on $\S$ joining $x_1$ and $x_3$.  
Since $a[x_1] = 0$ and $a[x_3] = \cc > r \ge a[x_2] + \eps$, and the fact that $a[\cdot]$ and $\gamma$ are continuous, there is $y \in \gamma$ such that $a[y] = a[x_2] + \eps$.  
Assume that $\kappa r \le 1$ so that $\cc \le 2r \le 2/\kappa$, which makes it possible to apply \lemref{shortest-arc} to obtain $\k[y] \le \kappa$, which in particular implies that $\theta[y] \ge \pi/2$.

\medskip\noindent{\em Case~2.1: Assume that $\theta \ge \pi/2$.}
Suppose that $\k[x_2] \ge (1+q) \kappa$ for some $q > 0$, for otherwise there is nothing to prove.  
By \lemref{triangle} below, and with the function $\phi$ defined there, we have
\begin{align*}
b[x_2] = \phi(a[x_2], c, \k[x_2])
&\ge \phi(a[y], c, \k[x_2]) \\
&\ge \phi(a[y], c, \kappa) + (\k[x_2] - \kappa) \partial_\kappa \phi(a[y], c, \kappa) \\
&\ge \phi(a[y], c, \k[y]) + (q\kappa) \kappa\, a[y]\, c\, (c-a[y])/4.
\end{align*}
In the 1st line we used the fact that $a[y] \ge a[x_2]$ and the monotonicity of $\phi$.
In the 2nd line we used the convexity of $\phi$.
In the 3rd line we used the fact that $\k[y] \le \kappa$ and the monotonicity of $\phi$, together with the inequality in \eqref{phi-deriv}. 
Noting that $b[y] = \phi(a[y], c, \k[y])$, we proved that 
\begin{align}
b[y] - b[x_2]
&\le - (q \kappa) \kappa\, a[y]\, c\, (c-a[y])/4 \notag \\
&\le - q\, \kappa^2\, a[x_2]\, c\, (c-a[x_2]-\eps)/4 \notag \\
&\le - \alpha q\, \kappa^2\, r^3/A_1. \label{b[y]-b[x2]}
\end{align}
In the 2nd line, we used $a[y] = a[x_2] + \eps$.
In the 3rd line, we used $a[x_2] \le c/\sqrt{2}$, which results from $a[x_2] \le b[x_2]$ and $\theta[x_2] \ge \pi/2$, together with $\alpha r \le a[x_2] \le c$, with $r < c \le 2r$, and we assumed that $\eps/r$ was small enough.

Let $x$ be a sample point such that $\|x-y\| \le \eps$.  
By the triangle inequality, 
\begin{align*}
a[x] &\le a[y] + \eps = a[x_2] + 2\eps \le r, \\
a[x] &\ge a[y] - \eps = a[x_2] \ge \alpha r,
\end{align*}
using the fact that $a[x_2] \le c/\sqrt{2} \le r - 3\eps$ (assuming $\eps/r$ is small enough), and by the same token,
\begin{align*}
b[x] &\le b[y] + \eps \le b[x_2] - \alpha q\, \kappa^2\, r^3/A_1  + \eps \le r - \alpha q\, \kappa^2\, r^3/A_1  + \eps, \\
b[x] &\ge b[y] - \eps \ge c - a[y] -\eps = c - a[x_2] - 2\eps \ge r (1 - 1/\sqrt{2}) - 2 \eps,
\end{align*}
by \eqref{b[y]-b[x2]}, and in the last inequality the fact that $a[x_2] \le c/\sqrt{2}$ and $c > r$.
Thus, if $q$ is large enough that $\alpha q\, \kappa^2\, r^3/A_1  \ge \eps$, and if $\eps/r$ is small enough, $(x_1,x,x_3)$ forms a path in the graph.  In addition, 
\beq\label{a[x]+b[x]}
a[x] + b[x] \le a[y] + b[y] + 2\eps \le a[x_2] + b[x_2] - \alpha q\, \kappa^2\, r^3/A_1  + 3\eps,
\eeq
by the triangle inequality first, and then \eqref{b[y]-b[x2]} and the fact that $a[y] = a[x_2] + \eps$ by construction.
Therefore, if $q$ were large enough that $\alpha q\, \kappa^2\, r^3/A_1  > 3\eps$, we would have $a[x] + b[x] < a[x_2] + b[x_2]$, which would contradict our working hypothesis that $(x_1,x_2,x_3)$ is a shortest path in the graph.  Therefore, we must have $\alpha q\, \kappa^2\, r^3/A_1  \le 3\eps$, meaning, $q \le 3A_1 \eps/(\alpha \kappa^2\, r^3)$.

\medskip\noindent{\em Case~2.2: Assume that $\theta < \pi/2$.}
Let $z$ be any point such that $a[z] = a[x_2]$ and $b[z] = \sqrt{\cc^2 -a[x_2]^2}$, so that $\angle(x_1,z,x_3) = \pi/2$.  
Redefine $q$ implicitly via $\k[z] = (1+q)\kappa$.  Since $\k[z] = 2/c$, explicitly, $q = 2/\kappa c - 1 \ge 1/\kappa r - 1$, since $c \le 2r$.  In particular, $q > 0$ as soon as $\kappa r \le 1/2$, which we assume henceforth.  Note that in this case $q \ge 1/(2\kappa r)$.
  
Replacing $x_2$ in Case~2.1 with $z$ --- which is possible because $z$ satisfies the same properties as $x_2$ in Case~2.1, except for the fact that $x_2$ is a sample point, but this is not used --- we find that there exists a point $x$ in the sample such that $\|x - z\| \le \eps$ and, as in \eqref{a[x]+b[x]}, satisfying
\beq
a[x] + b[x] 
\le a[z] + b[z] - \alpha q\, \kappa^2\, r^3/A_1  + 3\eps.
\eeq
Because $a[z] = a[x_2]$ and $b[z] \le b[x_2]$ by construction, this implies that
\beq
a[x] + b[x] 
\le a[x_2] + b[x_2] - \alpha q\, \kappa^2\, r^3/A_1  + 3\eps.
\eeq
However, here, 
\beq
-\alpha q\, \kappa^2\, r^3/A_1  + 3\eps 
\le -\alpha \kappa\, r^2/(2A_1)  + 3\eps 
< 0,
\eeq
whenever $\eps/\alpha \kappa r^2 < 1/(6A_1)$, and when this is the case, $a[x] + b[x] < a[x_2]+b[x_2]$, which is a contradiction since $x$ is a sample point and $a[\cdot] + b[\cdot]$ is assumed to be minimal at $x_2$ among sample points.
Hence, when $\kappa r$ and $\eps/\alpha \kappa r^2$ are both sufficiently small, we must have $\theta \ge \pi/2$, that is, we must be in Case~2.1.
\end{proof}

\begin{lem} \label{lem:triangle}
Consider a triangle with side lengths $a,b,c$ with $a^2 + b^2 \le c^2$.  Let $\kappa$ denote the inverse of its circumradius.  Then 
\beq\label{phi}
b = \phi(a,c,\kappa) := c \sqrt{1 - \tfrac14 \kappa^2 a^2} - a \sqrt{1 - \tfrac14 \kappa^2 c^2}.
\eeq
The function $\phi$ is decreasing in $a$ as well as increasing and convex in $\kappa$, with
\beq\label{phi-deriv}
\partial_\kappa \phi(a,c,\kappa) 
= \frac{\kappa a c}4 \left(\frac{c}{(1 - \tfrac14 \kappa^2 c^2)^{1/2}} -\frac{a}{(1 - \tfrac14 \kappa^2 a^2)^{1/2}}\right) 
\ge \frac{\kappa a c (c-a)}4. 
\eeq
\end{lem}

\begin{proof}
The expression \eqref{phi} is a simple consequence of the law of cosines, which says that 
\beq\label{law-cosines}
c^2 = a^2 + b^2 - 2 ab \cos\theta = a^2 + b^2 +2ab \sqrt{1 - \tfrac14 \kappa^2 c^2},
\eeq 
where $\theta$ is the angle opposite $c$, and we then use the expression $\kappa = (2 \sin\theta)/c$.  The monotonicity is elementary and the convexity comes from the fact that
\beq\label{phi-deriv2}
\partial_{\kappa\kappa} \phi(a,c,\kappa) = \frac{a c}4 \left(\frac{c}{(1 - \tfrac14 \kappa^2 c^2)^{3/2}} -\frac{a}{(1 - \tfrac14 \kappa^2 a^2)^{3/2}}\right) \ge 0,
\eeq
since $c \ge a$.
For the inequality in \eqref{phi-deriv}, we observe that $f(t) := t/(1-t^2)^{1/2}$ defined on $[0,1)$ has derivative $f'(t) = 1/(1-t^2)^{3/2} \ge 1$, and in particular is convex, so that
\begin{align*}
\frac{c}{(1 - \tfrac14 \kappa^2 c^2)^{1/2}} -\frac{a}{(1 - \tfrac14 \kappa^2 a^2)^{1/2}} 
&= \frac2\kappa (f(\kappa c/2) - f(\kappa a/2)) \\
&\ge \frac2\kappa f'(\kappa a/2) (\kappa c/2 - \kappa a/2) 
\ge c-a. \qedhere
\end{align*}
\end{proof}

\section*{Acknowledgments}

The authors wish to thank Stephanie Alexander, I.~David Berg, Richard Bishop, Dmitri Burago, Bruce Driver, and Bruno Pelletier for very helpful discussions.
The papers was carefully read by two anonymous referees, to which we are grateful. 
Some of the symbolic calculations were done with Wolfram$|$Alpha.\footnote{ \url{http://www.wolframalpha.com}}  
This work was partially supported by the US National Science Foundation (DMS 0915160, DMS 1513465).

\small
\bibliography{constrained-path}
\bibliographystyle{chicago}

\end{document}